\documentclass[12pt]{amsart}
\usepackage{amsmath,amsfonts,amsthm,amssymb,amscd,color}
\usepackage{graphicx,epsf}
\usepackage[bookmarks]{hyperref}
\newtheorem{theorem}{Theorem}

\newtheorem{remark}{Remark}

\newtheorem{proposition}{Proposition}
\def\R{\mathbb{R}}

\DeclareMathOperator{\FT}{\mathscr{F}}
\DeclareMathOperator{\iFT}{\mathscr{F}^{-1}}

\definecolor{KMpink}{RGB}{255, 153, 255}
\definecolor{KMgreen}{RGB}{0, 127, 0}

\usepackage{color}
\usepackage[usenames,dvipsnames]{xcolor}

\usepackage{amsmath,amsfonts,latexsym,mathrsfs}

\newlength{\dhatheight}

\newenvironment{proof1}{
    \noindent {\it Proof }}{\hfill$\Box$
}

\begin{document}

\title[Regularized Boussinesq--Ostrovsky equation]{\bf Validity of the weakly-nonlinear solution
of the Cauchy problem for the Boussinesq--Ostrovsky equation}

\author[K.R. Khusnutdinova, K.R. Moore, D.E. Pelinovsky]{K.R. Khusnutdinova$^{1}$, K.R. Moore$^{1}$, and D.E. Pelinovsky$^{2}$}

\address{$^{1}$ Department of Mathematical Sciences, Loughborough University, Loughborough, LE11 3TU, UK \\
$^{2}$ Department of Mathematics and Statistics, McMaster University, Hamilton ON, Canada, L8S 4K1 }

\date{\today}
\maketitle

\begin{abstract}
We consider the initial-value problem for the regularized Boussinesq--Ostrovsky
equation in the class of periodic functions. Validity of the weakly-nonlinear solution,
given in terms of two counter-propagating waves satisfying the uncoupled Ostrovsky equations,
is examined. We prove analytically and illustrate numerically that the improved accuracy of the solution
can be achieved at the time scales of the Ostrovsky equation if solutions of the linearized
Ostrovsky equations are incorporated into the asymptotic solution. Compared to the previous
literature, we show that the approximation error can be controlled in the energy space
of periodic functions and the nonzero mean values of the periodic functions can be
naturally incorporated in the justification analysis.
\end{abstract}

\section{Introduction}
\label{1}

Validity of the long-wave approximation for shallow water waves has been considered
in many recent works. Unbounded spatial domains and classes of decaying initial data were
typically considered. The first results in this direction were found in the context of water
waves by Craig \cite{Craig}, Schneider \cite{S}, Schneider \& Wayne \cite{SW1,SW2},
Ben Youssef \& Colin \cite{YC}, and Lannes \cite{Lannes}. Rigorous justification
analysis was developed to control the approximation error, and the bounds on the error terms were
typically found to be larger than those in the formal asymptotic theory.

Wayne \& Wright \cite{WW} extended this analysis to the regularized
Boussinesq equation to incorporate the first-order correction to the
leading-order approximation and to control the error term in Sobolev spaces
(see \cite{W} for a similar treatment of the original water wave equations).
They also reported numerical approximations that illustrated the validity of the main result,
where the bounds on the approximation error were controlled to be of the same size
as in the formal asymptotic theory.

Comprehensive treatment of the Boussinesq systems was developed by Bona {\em et al.} \cite{Bona},
who explored the case of symmetric Boussinesq systems and justified the long-wave approximation
(where the first-order correction term was added to the leading-order term)
in a number of physical models that included the water-wave equations
and the regularized Boussinesq system. A priori energy estimates
and Gronwall's inequalities were used to control the error term in Sobolev spaces
and to recover the error bounds of the formal asymptotic theory.
Initial data on the infinite line with sufficient decay at infinity were
treated on equal footing with the periodic initial data under
the zero mean-value assumption. The approximation error for the periodic
data with a nonzero mean value was found to be significantly larger.

Recently, in the framework of a system of coupled Boussinesq equations \cite{Karima1},
the long-wave approximation was extended to
the regularized Boussinesq--Ostrovsky equation
\begin{equation}
\label{Bous}
u_{\tau \tau} - u_{\xi \xi} + \delta u = \frac{1}{2} (u^2)_{\xi \xi} + u_{\tau \tau \xi \xi},
\end{equation}
where $\delta > 0$ and the subscripts denote partial differentiation.
The Boussinesq--Ostrovsky equation (\ref{Bous}) arises also in the context of
oceanic waves, which takes into account the effect of background rotation \cite{Gerkema}.
Therefore it is sometimes called the rotation modified Boussinesq equation.
This equation is a two-directional version of the Ostrovsky equation \cite{Ostrovsky},
which constitutes a modification of the Korteweg--de Vries (KdV) equation with an additional term
for $\delta \neq 0$. However, more recently the regularized Boussinesq--Ostrovsky equation
and a system of coupled Boussinesq equations have appeared in the context of a modified Toda lattice
on an elastic substrate  \cite{Kawahara} and  a layered solid waveguide with the soft
bonding layer \cite{Samsonov1,Samsonov2}.

A formal weakly-nonlinear solution of the initial-value problem for
a system of coupled Boussinesq equations on the infinite interval has been
constructed in terms of solutions of coupled and uncoupled Ostrovsky equations
for unidirectional counter-propagating waves in the recent work \cite{Karima1},
depending on the difference between the linear long-wave speeds of the waves.
Numerical simulations showed generation of a radiating solitary wave (in the
case of strong interactions) and strongly nonlinear wave packets (in the case of weak interactions)
from localized initial data.  Radiating solitary waves in coupled KdV equations
were previously observed in \cite{FDG}. The emergence of the strongly nonlinear
wave packets for the Ostrovsky equation was reported in \cite{GH}. A discrete version
of the same phenomenon was studied in \cite{Kawahara}.

Only the weak interaction case related to two uncoupled Ostrovsky equations
is relevant in the case of the scalar Boussinesq--Ostrovsky equation (\ref{Bous}).
Explicit analytical solutions and more detailed numerical simulations were developed in the follow-up work
\cite{Karima2} in the context of the regularized Boussinesq equation (with $\delta = 0$).
In particular, the explicit weakly nonlinear solution of the regularized Boussinesq
equation was constructed for the initial data corresponding to the soliton solution
of the KdV equation. The solution shows  generation of a small counter-propagating
solitary wave, which agrees with the numerical simulations. Explicit analytical solutions
have also been constructed for the initial data in the form of the $N$-soliton solutions
of the KdV equation and their perturbations.
 It was shown in \cite{Karima2} that the error term is significantly smaller when the first-order
correction term is taken into account, but it grows with the time.  No detailed studies
of the convergence rate for the error term were reported in \cite{Karima1,Karima2}.

The purpose of this work is to develop a systematic approach to the derivation
and justification of the error terms of the weakly-nonlinear solutions
for the regularized Boussinesq--Ostrovsky equation when the first-order
correction is added to the leading-order term. We will also give
a systematic comparison of the convergence rates predicted by the theory
and observed in numerical simulations.

Our main results are obtained in the periodic domain, where derivation and justification of the reduced
equations become easy with the use of Fourier series (see \cite{PSW} for similar derivations).
However, the infinite spatial domain can be obtained in the limit to the
infinite period. Therefore, we can include numerical examples that resemble solitary waves.

The novelties of our paper (compared to the results obtained in \cite{Bona}, \cite{Karima1,Karima2}, and \cite{WW})
are the following. First, we discuss in details the role of the nonzero mean value of
the periodic solutions in the justification analysis.
For $\delta = 0$, nonzero mean values were found to degrade the accuracy of the
long-wave approximation \cite{Bona} because no corrections to the wave speeds were
introduced in the uncoupled KdV equations. For $\delta \neq 0$, all solutions of the Ostrovsky equations must
satisfy the zero-mean constraint \cite{Benilov} but the zero-mean constraint does not have
to be assumed for the Boussinesq--Ostrovsky equation. We will show how the mean value can be incorporated
in the long-wave approximation of the regularized Boussinesq--Ostrovsky equation (\ref{Bous}).

The other novelty of our work is that we extend the analysis
to justify the first-order correction terms to the long-wave
approximation on the time scale of the Ostrovsky equation. The first-order correction
terms were written explicitly in the previous works \cite{Karima1,Karima2} but
the validity of these terms at times greater than order one can only be achieved if these terms satisfy
the linearized Ostrovsky equation. The linearized KdV equations
have also appeared in \cite{WW} in the framework of the regularized Boussinesq equation
(with $\delta = 0$). Compared to this work, we give precise expressions in terms of
solutions of the leading order equations to solve the associated
initial-value problem up to the first-order corrections terms and develop the justification
analysis in the energy space of the regularized Boussinesq equation.

Before closing the introduction, we shall also discuss the reductive
perturbation schemes, which were used to obtain the integrable KdV hierarchy
from the shallow water-wave and Boussinesq equations \cite{Manna1,Manna2,Manna3}.
It was later shown in \cite{Kodama1,Kodama2}
that there are obstacles to the asymptotic integrability of the original physical equations
reduced to the integrable KdV hierarchy in the sense that the formal asymptotic expansions
become non-uniform at higher orders of $\epsilon$. Within the framework of our approach,
we do not need to set up the time evolution along higher flows of the KdV hierarchy
if we are only interested in the validity of the first-order correction terms
at the time scale of the KdV equation. In other words, we can fully
control the approximation error within the required order of accuracy of the
asymptotic expansions without analyzing the secular terms in the linearized
KdV equations and the related difference between asymptotically
integrable and non-integrable perturbation terms.

This paper is organized as follows. In Section 2, we set up the long-wave scaling
and analyze dynamics of the nonzero mean value of periodic functions. In Section 3, we describe
the formal asymptotic theory and prove the justification result about the approximation error
of the asymptotic expansion, in the framework of the regularized Boussinesq equation
(with $\delta = 0$). Section 4 extends analysis to the case of the regularized
Boussinesq--Ostrovsky equation (with $\delta > 0$). Section 5 illustrates the main results
with numerical computations. Section 6 concludes the paper.

\section{Long-wave scaling and dynamics of the mean value}

Using the scaling transformation,
$$
u(x,t) = \epsilon U(x,t), \quad x = \sqrt{\epsilon} \xi,
\quad t =  \sqrt{\epsilon} \tau, \quad \delta = \epsilon^2 \gamma, \quad \epsilon > 0,
$$
we rewrite the regularized Boussinesq--Ostrovsky equation (\ref{Bous}) in the equivalent form:
\begin{equation}
\label{BO}
U_{tt} - U_{xx} = \epsilon \left( \frac{1}{2} (U^2)_{xx} + U_{tt xx} - \gamma U \right),
\end{equation}
Throughout this paper, we will interpret the right-hand side as a $\mathcal{O}(\epsilon)$ perturbation
and develop the asymptotic theory in the limit of small $\epsilon$.
In physical settings \cite{Manna1}, the perturbed wave equation (\ref{BO}) has also the
$\mathcal{O}(\epsilon^2)$ correction terms, which need to be incorporated in the asymptotic
theory. However, these modifications are technically straightforward and obey the same
justification analysis, hence we are going to neglect the $\mathcal{O}(\epsilon^2)$
terms in the Boussinesq equation (\ref{BO}).

We consider the initial-value problem with the initial data
\begin{equation}
\label{initial-values}
U |_{t = 0} = F(x), \quad U_t |_{t = 0} = V(x),
\end{equation}
where the given functions $F$ and $V$ are in the class of squared integrable 
$(2L)$-periodic functions. We can expand them into the Fourier series
\begin{equation}
\label{FS-IO}
F(x) = \sum_{n \in \mathbb{Z}} F_n e^{\frac{i \pi n x}{L}}, \quad
V(x) = \sum_{n \in \mathbb{Z}} V_n e^{\frac{i \pi n x}{L}}.
\end{equation}
For the sake of simplification, we assume that $F$ and $V$ are $\epsilon$-independent,
although extension to a general case is also straightforward.

The following local existence result is similar to the local well-posedness theory for
regularized Boussinesq systems \cite{Chen1,Chen2}.

\begin{proposition}
Fix $s > \frac{1}{2}$. For any $(F,V) \in H^s_{\rm per}(-L,L) \times H^{s}_{\rm per}(-L,L)$, 
there exists an $\epsilon$-independent $t_0 > 0$ and a unique solution
$U(t) \in C^1([0,t_0],H^s_{\rm per}(-L,L))$ of the modified
regularized Boussinesq equation (\ref{BO}) with any $\epsilon > 0$ and
$\gamma \geq 0$. \label{proposition-local}
\end{proposition}

\begin{proof}
The evolution problem can be written in the operator form:
\begin{equation}
\label{BO-operator}
U_{tt} - L_{\epsilon} U_{xx} + \epsilon \gamma L_{\epsilon} U = M_{\epsilon} U^2,
\end{equation}
where
$$
L_{\epsilon} := (1 - \epsilon \partial_x^2)^{-1}, \quad
M_{\epsilon} := \frac{1}{2} \epsilon \partial_x^2  L_{\epsilon}.
$$
By using Fourier series, we realize that both
operators $L_{\epsilon}$ and $M_{\epsilon}$
are bounded for any $\epsilon > 0$ and $\gamma \geq 0$
with the $\epsilon$-independent bounds:
$$
\| L_{\epsilon} U \|_{L^2_{\rm per}} \leq \| U \|_{L^2_{\rm per}}, \quad
\| M_{\epsilon} U \|_{L^2_{\rm per}} \leq \frac{1}{2} \| U \|_{L^2_{\rm per}}.
$$
Using Duhamel's principle, we rewrite the evolution problem (\ref{BO-operator})
with the initial data (\ref{initial-values}) in the equivalent integral form:
\begin{equation}
\label{BO-operator-solution}
U(t) = S_t(t) \star F + S(t) \star V + \int_0^t S(t - \tau) \star
\left( M_{\epsilon} U^2(\tau) - \epsilon \gamma L_{\epsilon} U(\tau) \right) d \tau,
\end{equation}
where the star denotes the convolution operator and $S(t)$ is the fundamental
solution operator with the Fourier image:
$$
\hat{S}(t) = \frac{\sin(k \hat{\ell}(k) t)}{k \hat{\ell}(k)}, \quad \ell(k) := \frac{1}{\sqrt{1 + \epsilon k^2}}.
$$
Because $S(t)$ and $S_t(t)$ are bounded operators from $L^2_{\rm per}(-L,L)$ to $L^2_{\rm per}(-L,L)$
for any $t \in \mathbb{R}$, the fixed-point iteration method (see, e.g., \cite{Cazenava}) implies that
there exists a unique local solution of the integral equation (\ref{BO-operator}) in the class
\begin{equation}
\label{BO-class}
U(t) \in C([0,t_0],H^s_{\rm per}(-L,L))
\end{equation}
for any $(F,V) \in H^s_{\rm per}(-L,L) \times H^{s}_{\rm per}(-L,L)$
and any $s > \frac{1}{2}$, where $t_0 > 0$ is an $\epsilon$-independent
local existence time. On the other hand, $U_t(t)$ is defined
by differentiation of the integral equation (\ref{BO-operator}):
$$
U_t(t) = S_{tt}(t) \star F + S_t(t) \star V + \int_0^t S_t(t - \tau) \star
\left( M_{\epsilon} U^2(\tau) - \epsilon \gamma L_{\epsilon} U(\tau) \right) d \tau.
$$
Because $S_{tt}(t)$ satisfies the bound
$$
\| S_{tt}(t) \star F \|_{L^2_{\rm per}} \leq \frac{1}{\sqrt{\epsilon}} \| F \|_{L^2_{\rm per}},
\quad \epsilon > 0
$$
and $U(t)$ is defined in the class (\ref{BO-class}), we have
$U_t(t) \in C([0,t_0],H^s_{\rm per}(-L,L))$ for any $\epsilon > 0$.
As a result, $U(t) \in C^1([0,t_0],H^s_{\rm per}(-L,L))$.
\end{proof}

In the long-wave approximation, we will need to extend local solutions
of the Boussinesq equation (\ref{BO}) in the energy space
to the time intervals with $t_0 = \mathcal{O}(\epsilon^{-1})$.
This continuation is achieved with energy methods resulting in the following wave breaking criterion.

\begin{proposition}
Let $U(t) \in C^1([0,t_0],H^1_{\rm per}(-L,L))$ be a local solution in Proposition \ref{proposition-local}.
The solution is extended to the time interval $[0,t_0']$ with $t_0' > t_0$ if
\begin{equation}
\label{breaking-criterion}
M := \sup_{t \in [0,t_0']} \| U(t) \|_{L^{\infty}_{\rm per}} +
\sup_{t \in [0,t_0']} \| U_t(t) \|_{L^{\infty}_{\rm per}} < \infty.
\end{equation}
\label{proposition-continuation}
\end{proposition}

\begin{proof}
Let us define the energy function
\begin{equation}
\label{energy-function}
E(U) := \int_{-L}^L \left( U_t^2 + U_x^2 + \epsilon \gamma U^2 + \epsilon U_{tx}^2 + \epsilon U U_x^2 \right) dx,
\end{equation}
for any local solution $U(t) \in C^1([0,t_0],H^1_{\rm per}(-L,L))$. Multiplying the Boussinesq equation (\ref{BO})
by $U_t$ and integrating by parts, we obtain
$$
\frac{d E(U)}{d t} = \epsilon \int_{-L}^L U_t U_x^2 dx + 2 \left( U_t U_x + \epsilon U_t U_{ttx}
+ \epsilon U U_t U_x \right) \biggr|_{x = -L}^{x = L}.
$$
By standard approximation arguments in Sobolev space $H^2_{\rm per}(-L,L)$, the trace of the boundary values
is zero and we obtain a priori energy estimate
$$
\frac{d E(U)}{d t} = \epsilon \int_{-L}^L U_t U_x^2 dx
\leq \epsilon \| U_t \|_{L^{\infty}_{\rm per}} \| U_x \|_{L^2_{\rm per}}^2.
$$
Under the condition (\ref{breaking-criterion}), there is $M$-dependent constant $C(M) > 0$
such that
$$
\| U_x \|_{L^2_{\rm per}}^2 \leq C(M) E(U).
$$
By Gronwall's inequality, we then
obtain
$$
E(U) \leq E(U_0) e^{\epsilon M C(M) t}, \quad t \in [0,t_0'],
$$
such that the solution is extended to the time $t_0' > t_0$ in the energy space,
that is, in the class $U(t) \in C^1([0,t_0'],H^1_{\rm per}(-L,L))$.
\end{proof}

\begin{remark}
By Sobolev embedding of $H^1_{\rm per}(-L,L)$ to $L^{\infty}_{\rm per}(-L,L)$,
we have $M = \mathcal{O}(\epsilon^{-1/2})$ and $C(M) = \mathcal{O}(1)$ as $\epsilon \to 0$. The energy method
of Proposition \ref{proposition-continuation} guarantees continuation of the
local solution of Proposition \ref{proposition-local} to the time intervals of
$t_0 = \mathcal{O}(\epsilon^{-1/2})$. However, this is not sufficient as the long-wave
approximation requires us to continue the solution to the time intervals of
$t_0 = \mathcal{O}(\epsilon^{-1})$. We achieve this goal by controlling $M$
with the $\mathcal{O}(1)$ bound as $\epsilon \to 0$
(see the proof of Theorem \ref{theorem-main} below).
\end{remark}

\begin{remark}
D'Alembert solution of the wave equation $U_{tt} - U_{xx} = 0$ (for $\epsilon = 0$)
only requires us to set
$$
(F,V) \in H^1_{\rm per}(-L,L) \times L^2_{\rm per}(-L,L)
$$
to have $U \in C(\mathbb{R},H^1_{\rm per}(-L,L)) \cap C^1(\mathbb{R},L^2_{\rm per}(-L,L))$. However,
we actually need to find a solution in the class $U \in C^1(\mathbb{R},H^1_{\rm per}(-L,L))$
in order to bound all terms in the energy norm (\ref{energy-function}).
This is achieved in Proposition \ref{proposition-local}, which gives
an improved local-wellposedness result for the regularized Boussinesq--Ostrovsky equation
(\ref{BO}).
\end{remark}

We shall now study the dynamics of the mean value of the $(2L)$-periodic solution $U(t) \in C^1([0,t_0],H^1_{\rm per}(-L,L))$.
Integrating equation (\ref{BO}) in $x$ over the period $(2L)$ in the class of sufficiently
smooth $(2L)$-periodic functions, we obtain the evolution equation for the mean value:
\begin{equation}
\label{balance-equation}
\frac{d^2}{d t^2} \int_{-L}^L U(x,t) dx = -\epsilon \gamma \int_{-L}^L U(x,t) dx,
\end{equation}
which shows that
\begin{equation}
\label{mean-value}
\langle U \rangle(t) := \frac{1}{2L} \int_{-L}^L U(x,t) dx =
F_0 \cos(\sqrt{\epsilon \gamma} t) +
V_0 \frac{\sin(\sqrt{\epsilon \gamma} t)}{\sqrt{\epsilon \gamma}},
\end{equation}
where $F_0$ and $V_0$ are mean values of the Fourier series (\ref{FS-IO}).

To eliminate the linear growth of the mean value $\langle U \rangle$ in $t$ for $\gamma = 0$, we
should assume that $V_0 = 0$ (that is, $V$ has zero mean value). If $\gamma > 0$,
the mean value $\langle U \rangle$ is oscillating with the frequency $\omega = (\epsilon \gamma)^{1/2}$,
and hence $F_0$ and $V_0$ can be nonzero in general. However, for any $t \in \mathbb{R}$,
the mean value $\langle U \rangle$ diverges as $\mathcal{O}(\epsilon^{-1/2})$ if $V_0 \neq 0$.
Therefore, in both cases $\gamma = 0$ and $\gamma > 0$, we would like to eliminate
the secular growth of the mean value $\langle U \rangle$ in $t$ by requiring that
\begin{equation}
\label{velocity-constraint}
V_0 = \frac{1}{2L} \int_{-L}^L V(x) dx = 0.
\end{equation}
In this case, the mean value $\langle U \rangle$ is bounded in $t \in \R$
and $\epsilon \in \R_+$ with a uniform limit as $\epsilon \to 0$
for any $\gamma \geq 0$.

\section{Long-wave approximation for $\gamma = 0$}

We shall consider the initial-value problem for the regularized Boussinesq equation,
\begin{equation}
\label{B}
U_{tt} - U_{xx} = \epsilon \left( \frac{1}{2} (U^2)_{xx} + U_{tt xx} \right),
\end{equation}
The initial data are given by (\ref{initial-values}) and (\ref{FS-IO})
subject to the zero-mean velocity constraint (\ref{velocity-constraint}).
By the exact solution (\ref{mean-value}), the mean value of the solution $U$ is constant
in $t$ with $\langle U \rangle = F_0$.

Substituting $U(x,t) = c_0 + \tilde{U}(x,t)$ into the regularized Boussinesq equation
(\ref{B}), where $c_0 := F_0$ and $\tilde{U}$ is the zero-mean part of the $2L$-periodic function $U$,
we obtain the evolution equation
\begin{equation}
\label{B-periodic}
\tilde{U}_{tt} - \tilde{U}_{xx} = \epsilon \left( c_0 \tilde{U}_{xx} +
\frac{1}{2} (\tilde{U}^2)_{xx} + \tilde{U}_{tt xx} \right).
\end{equation}
By Proposition \ref{proposition-local},
there exists a unique local solution $\tilde{U} \in C^1([0,t_0],H^s_{\rm per}(-L,L))$
of the evolution equation (\ref{B-periodic})
for any $(\tilde{F},\tilde{V}) \in H^s_{\rm per}(-L,L) \times H^{s}_{\rm per}(-L,L)$
with $s > \frac{1}{2}$, where $t_0 > 0$ is a local existence time and
the tilded variables denote the zero-mean part of the $(2L)$-periodic functions.

\subsection{Derivation}

We shall look for a formal asymptotic solution of the evolution equation (\ref{B-periodic})
up to and including the $\mathcal{O}(\epsilon^2)$ terms:
\begin{equation}
\label{asymptotic-expansion}
\tilde{U}(x,t) = U_0(x,t) + \epsilon U_1(x,t) + \epsilon^2 U_2(x,t) + \mathcal{O}(\epsilon^3).
\end{equation}
In the formal theory, we collect together terms at each order.

\vspace{0.2cm}

{\bf Order $\mathcal{O}(\epsilon^0)$:} The leading order $U_0$ satisfies the initial-value problem
for the wave equation:
\begin{equation}
\label{zero-order}
\left\{ \begin{array}{l} (\partial_t^2 - \partial_x^2) U_0 = 0, \\
U_0 |_{t = 0} = \tilde{F}, \\
\partial_t U_0 |_{t = 0} = \tilde{V}. \end{array} \right.
\end{equation}
The d'Alembert solution of the wave equation (\ref{zero-order})
consists of a superposition of two counter-propagating waves of zero mean:
\begin{equation}
\label{leading-order-before}
U_0(x,t) = f^-(\xi_-) + f^+(\xi_+), \quad \xi_{\pm} = x \pm t,
\end{equation}
where
\begin{equation}
\label{leading-order-before-ic}
f^{\pm}(\xi_{\pm}) = \frac{1}{2} \tilde{F}(\xi_{\pm}) \pm \frac{1}{2} \partial_{\xi_{\pm}}^{-1} \tilde{V}(s) =
\frac{1}{2} \sum_{n \in \mathbb{Z} \backslash \{0\}} \left( F_n \pm \frac{L V_n}{\pi i n} \right)
e^{\frac{i \pi n \xi_{\pm}}{L}},
\end{equation}
where $\partial_{\xi_{\pm}}^{-1}$ denote the zero-mean anti-derivative of the zero-mean
periodic functions.

\vspace{0.2cm}

{\bf Order $\mathcal{O}(\epsilon)$:} At this point, we should realize that the next-order correction
terms are going to grow linearly in time $t$ unless we will modify the leading-order solution
on a slow time scale $T = \epsilon t$. Therefore, we modify the leading-order
solution (\ref{leading-order-before}) with the slow time variable:
\begin{equation}
\label{leading-order}
U_0(x,t) = f^-(\xi_-,T) + f^+(\xi_+,T), \quad f^{\pm}(\xi_{\pm},T) =
\sum_{n \in \mathbb{Z} \backslash \{0\}} a^{\pm}_n(T) e^{\frac{i \pi n \xi_{\pm}}{L}},
\end{equation}
where
\begin{equation}
\label{ic-kdv}
a^{\pm}_n |_{T = 0} = \frac{1}{2} \left( F_n \pm \frac{L V_n}{\pi i n} \right),
\quad n \in \mathbb{Z} \backslash \{0\}.
\end{equation}

We know the initial data for $f^{\pm}$ in slow time $T = \epsilon t$, but we do not
know yet the evolution equations for $f^{\pm}(\xi_{\pm},T)$.
To derive these equations, we consider the first-order correction terms:
\begin{equation}
\label{first-order}
 \left\{ \begin{array}{l} (\partial_t^2 - \partial_x^2) U_1 =
-2 \partial^2_{t T} U_0 + c_0 \partial_{x}^2 U_0
+ \frac{1}{2} \partial_x^2(U_0^2) + \partial^4_{ttxx} U_0, \\
U_1 |_{t = 0} = 0, \\
\partial_t U_1 |_{t = 0} = - \partial_T U_0 |_{t = 0}. \end{array} \right.
\end{equation}
Using the Fourier series, $U_1(x,t) = \sum_{n \in \mathbb{Z} \backslash \{0\}}
g_n(t) e^{\frac{i \pi n x}{L}}$, we reduce the evolution equation
in the system (\ref{first-order}) to the uncoupled system
of differential equations:
\begin{equation}
\label{first-order-ode}
\frac{d^2 g_n}{d t^2} + \left( \frac{\pi n}{L} \right)^2 g_n = h_n(t),
\end{equation}
where
\begin{eqnarray*}
h_n(t) & = & -\frac{2 \pi i n}{L} \left( \frac{d a^+_n}{d T} e^{\frac{i \pi n t}{L}} - \frac{d a^-_n}{d T}
e^{-\frac{i \pi n t}{L}} \right) + \frac{\pi^4 n^4}{L^4} \left( a^+_n e^{\frac{i \pi n t}{L}} + a^-_n
e^{-\frac{i \pi n t}{L}} \right) \\
& \phantom{t} & - \frac{\pi^2 n^2}{L^2} c_0 \left( a^+_n e^{\frac{i \pi n t}{L}} + a^-_n
e^{-\frac{i \pi n t}{L}} \right)
- \frac{\pi^2 n^2}{L^2} \left( \sum_{k \in \mathbb{Z}\backslash \{0,n\}} a^+_k a^-_{n-k} e^{\frac{i \pi (2k-n) t}{L}} \right)\\
& \phantom{t} & -\frac{\pi^2 n^2}{2 L^2}
\left( \sum_{k \in \mathbb{Z} \backslash \{0,n\}} a^+_k a^+_{n-k} \right) e^{\frac{i \pi n t}{L}}
-\frac{\pi^2 n^2}{2 L^2} \left( \sum_{k \in \mathbb{Z}\backslash \{0,n\}} a^-_k a^-_{n-k} \right) e^{-\frac{i \pi n t}{L}}.
\end{eqnarray*}

The terms $e^{\pm \frac{i \pi n t}{L}}$ in the right-hand-side of differential equations
(\ref{first-order-ode}) are resonant (that is, they induce a linear growth of $g_n(t)$ in a fast time $t$).
To remove these resonant terms, we define uniquely the evolution problem for the Fourier coefficients
of the leading-order solution (\ref{leading-order}):
\begin{eqnarray}
\label{constraints}
\mp \frac{2 \pi i n}{L} \frac{d a^{\pm}_n}{d T}
+ \frac{\pi^4 n^4}{L^4} a^{\pm}_n - \frac{\pi^2 n^2}{L^2} c_0 a^{\pm}_n
- \frac{\pi^2 n^2}{2 L^2} \sum_{k \in \mathbb{Z} \backslash \{0,n\}} a^{\pm}_k a^{\pm}_{n-k} = 0,
\end{eqnarray}
subject to the initial conditions (\ref{ic-kdv}).

In the equivalent differential form, the evolution problem (\ref{constraints})
coincides with the two uncoupled KdV equations
\begin{equation}
\label{KdV}
\frac{\partial}{\partial \xi_{\pm}} \left( \mp 2 \frac{\partial f^{\pm}}{\partial T}
+ \frac{\partial^3 f^{\pm}}{\partial \xi_{\pm}^3}
+ c_0 \frac{\partial f^{\pm}}{\partial \xi_{\pm}}
+ f^{\pm} \frac{\partial f^{\pm}}{\partial \xi_{\pm}} \right) = 0.
\end{equation}

We consider the initial-value problem for the uncoupled KdV
equations (\ref{KdV}) starting with the initial values $f^{\pm} |_{T = 0}$
given by (\ref{leading-order-before-ic}).
By the local and global well-posedness theory for the KdV equation \cite{CKSTT},
there exist unique global solutions $f^{\pm} \in C(\mathbb{R}_+,H^s_{\rm per}(-L,L))$
of the KdV equations (\ref{KdV}) for any $f^{\pm} |_{T = 0} \in H^s_{\rm per}(-L,L)$
with $s \geq -\frac{1}{2}$.

After the constraints (\ref{constraints}) are substituted back to the differential equations
(\ref{first-order-ode}), we obtain the linear inhomogeneous equations
\begin{eqnarray*}
\frac{d^2 g_n}{d t^2} + \left( \frac{\pi n}{L} \right)^2 g_n =
-\frac{\pi^2 n^2}{L^2} \sum_{k \in \mathbb{Z}\backslash\{0,n\}} a^+_k a^-_{n-k} e^{\frac{i \pi (2k-n) t}{L}},
%\label{first-order-ode-ivc}
\end{eqnarray*}
subject to the initial conditions
\begin{equation*}
%\label{ivc-first-order}
g_n(0) = 0, \quad \partial_t g_n(0) = - \partial_T a_n^+(0) - \partial_T a_n^-(0).
\end{equation*}
This initial-value problem admits the following bounded solution:
$$
g_n(t) = \sum_{k \in \mathbb{Z} \backslash \{0,n\}} \frac{n^2}{4 k (k-n)} a_k^+ a_{n-k}^-
e^{\frac{i \pi (2k-n) t}{L}} + G_n \cos\left(\frac{\pi n t}{L}\right) + H_n  \sin\left(\frac{\pi n t}{L}\right),
$$
where $G_n$ and $H_n$ are constants of integrations to be found from the initial conditions for $g_n$.
Using the previous expression for $g_n$, we rewrite the first-order correction term in the implicit form:
\begin{equation}
\label{first-order-correction}
U_1(x,t) = f_c(x,t) + \phi^-(\xi_-,T) + \phi^+(\xi_+,T),
\end{equation}
where $f_c$ is uniquely defined by
\begin{equation}
\label{f-c-Fourier}
f_c(x,t) = \sum_{n \in \mathbb{Z}\backslash \{0\}} \sum_{k \in \mathbb{Z} \backslash \{0,n\}} A_{n,k}(T)
e^{\frac{i \pi (2k-n) t}{L} + \frac{i \pi n x}{L}},
\end{equation}
where
\begin{equation*}
A_{n,k}(T) := \frac{n^2}{4 k (k-n)} a_k^+(T) a_{n-k}^-(T),
\end{equation*}
and $\phi^{\pm}$ are counter-propagating waves of the wave equation given by
\begin{eqnarray}
\label{phi-pm-Fourier}
\phi^{\pm}(\xi_{\pm},T) = \sum_{n \in \mathbb{Z} \backslash \{0\}} b^{\pm}_n(T) e^{\frac{i \pi n \xi_{\pm}}{L}}.
\end{eqnarray}

An explicit expression for the first-order correction term was derived recently in \cite{Karima1,Karima2}
by averaging the differential equations with respect to the fast time in characteristic coordinates. We check that
our expression (\ref{f-c-Fourier}) coincides with the one derived in \cite{Karima1,Karima2}:
\begin{eqnarray}
f_c(x,t) =  -\frac{1}{4} \left( 2 f^+ f^- + (\partial_{\xi^+} f^+) (\partial_{\xi_-}^{-1} f^-)
+ (\partial_{\xi^-} f^-) (\partial_{\xi_+}^{-1} f^+)\right), \label{expression}
\end{eqnarray}
where $\partial_{\xi_{\pm}}^{-1} f^{\pm}$ denote again the zero-mean anti-derivatives of the zero-mean periodic functions
$f^{\pm}$.

Using the initial conditions for $g_n$, we can express the initial data
for the amplitudes $b_n^{\pm}$ explicitly:
\begin{equation}
\label{ic-kdv-lin}
b^{\pm}_n |_{T = 0} = - \sum_{k \in \mathbb{Z} \backslash \{0,n\}}
\frac{n (n \pm (2k-n))}{8 k (k-n)} \left( a_k^+ a_{n-k}^- \right)|_{T = 0}
\mp \frac{L}{2 i \pi n} \left( \frac{d a_n^+}{d T} + \frac{d a_n^-}{d T}\right) \biggr|_{T = 0}
\end{equation}

\vspace{0.2cm}

{\bf Order $\mathcal{O}(\epsilon^2)$:} The time evolution of $\phi^{\pm}$ with respect to the slow time $T$
is not defined at the $\mathcal{O}(\epsilon)$ order. To derive the time evolution, we consider
the second-order correction terms:
\begin{equation}
\label{second-order}
 \left\{ \begin{array}{l} (\partial_t^2 - \partial_x^2) U_2 =
-2 \partial^2_{t T} U_1 - \partial_T^2 U_0 + c_0 \partial_x^2 U_1 + \partial_x^2(U_0 U_1) +
\partial^4_{ttxx} U_1 + 2\partial^4_{tTxx} U_0, \\
U_2 |_{t = 0} = 0, \\
\partial_t U_2 |_{t = 0} = - \partial_T U_1 |_{t = 0}. \end{array} \right.
\end{equation}

Using the Fourier series again $U_2(x,t) = \sum_{n \in \mathbb{Z} \backslash \{0\}} g_n(t) e^{\frac{i \pi n x}{L}}$,
we reduce the evolution equation in the system (\ref{second-order}) to the uncoupled system
of differential equations:
\begin{equation}
\label{second-order-ode}
\frac{d^2 g_n}{d t^2} + \left( \frac{\pi n}{L} \right)^2 g_n = h_n(t),
\end{equation}
where
\begin{eqnarray*}
h_n(t) & = & -\frac{d^2 a^+_n}{d T^2} e^{\frac{i \pi n t}{L}} - \frac{d^2 a^-_n}{d T^2}
e^{-\frac{i \pi n t}{L}} -\frac{2 \pi i n}{L} \left( \frac{d b^+_n}{d T} e^{\frac{i \pi n t}{L}} - \frac{d b^-_n}{d T}
e^{-\frac{i \pi n t}{L}} \right) \\
& \phantom{t} & - \sum_{k \in \mathbb{Z} \backslash \{0,n\}} \frac{2 i \pi (2k-n)}{L} \frac{d A_{n,k}}{d T}
e^{\frac{i \pi (2k-n) t}{L}} + \frac{\pi^4 n^4}{L^4} \left( b^+_n e^{\frac{i \pi n t}{L}} + b^-_n
e^{-\frac{i \pi n t}{L}} \right) \\
& \phantom{t} & + \sum_{k \in \mathbb{Z} \backslash \{0,n\}} \frac{\pi^4 n^2 (2k-n)^2}{L^4} A_{n,k}
e^{\frac{i \pi (2k-n) t}{L}} - \frac{2 i \pi^3 n^3}{L^3} \left( \frac{d a^+_n}{d T} e^{\frac{i \pi n t}{L}} -
\frac{d a^-_n}{d T} e^{-\frac{i \pi n t}{L}} \right) \\
& \phantom{t} & -\frac{\pi^2 n^2}{L^2} \sum_{k \in \mathbb{Z} \backslash \{0\}} a^+_k
\left( b^+_{n-k} e^{\frac{i \pi n t}{L}} + b^-_{n-k} e^{\frac{i \pi (2k-n) t}{L}} +
\sum_{l \in \mathbb{Z} \backslash \{0,n-k\}} A_{n-k,l} e^{\frac{i \pi (2l + 2k - n) t}{L}} \right) \\
& \phantom{t} & -\frac{\pi^2 n^2}{L^2} \sum_{k \in \mathbb{Z} \backslash \{0\}} a^-_k
\left( b^+_{n-k} e^{\frac{i \pi (n-2k) t}{L}} + b^-_{n-k} e^{-\frac{i \pi n t}{L}} +
\sum_{l \in \mathbb{Z} \backslash \{0,n-k\}} A_{n-k,l} e^{\frac{i \pi (2l - n) t}{L}} \right) \\
& \phantom{t} & - \frac{\pi^2 n^2}{L^2} c_0 \sum_{k \in \mathbb{Z} \backslash \{0,n\}}A_{n,k}
e^{\frac{i \pi (2k-n) t}{L}} - \frac{\pi^2 n^2}{L^2} c_0 \left( b^+_n e^{\frac{i \pi n t}{L}} + b^-_n
e^{-\frac{i \pi n t}{L}} \right).
\end{eqnarray*}

The terms $e^{\pm \frac{i \pi n t}{L}}$ in the right-hand-side of differential equations
(\ref{second-order-ode}) are again resonant.
To remove these terms, we define uniquely the evolution problem for the Fourier coefficients
of the first-order solution (\ref{first-order-correction}):
\begin{eqnarray}
\label{constraints-second-order}
- \frac{d^2 a_n^{\pm}}{d T^2} \mp \frac{2 \pi i n}{L} \frac{d b^{\pm}_n}{d T}
+ \frac{\pi^4 n^4}{L^4} b^{\pm}_n \mp \frac{2i \pi^3 n^3}{L^3} \frac{d a_n^{\pm}}{d T}
- \frac{\pi^2 n^2}{L^2} c_0  b^{\pm}_n \\
\nonumber
- \frac{\pi^2 n^2}{L^2} \sum_{k \in \mathbb{Z} \backslash \{0,n\}} a^{\pm}_k b^{\pm}_{n-k}
- \frac{\pi^2 n^2}{L^2} \sum_{k \in \mathbb{Z} \backslash \{0,n\}} \frac{(n-k)^2}{4 n k} a^{\pm}_n |a^{\mp}_{k}|^2 = 0,
\end{eqnarray}
subject to the initial conditions (\ref{ic-kdv-lin}).

In the equivalent differential form, the evolution equations (\ref{constraints-second-order})
coincide with the linearized KdV equations
\begin{equation}
\label{KdV-lin}
\frac{\partial}{\partial \xi_{\pm}} \left( \mp 2 \frac{\partial \phi^{\pm}}{\partial T}
+ \frac{\partial^3 \phi^{\pm}}{\partial \xi_{\pm}^3} + c_0 \frac{\partial \phi^{\pm}}{\partial \xi_{\pm}}
+ \frac{\partial}{\partial \xi_{\pm}} f^{\pm} \phi^{\pm} \right) =
\frac{\partial^2 f^{\pm}}{\partial T^2} \mp 2 \frac{\partial^4 f^{\pm}}{\partial \xi_{\pm}^3 T}
+ \frac{\partial^2 f^{\pm}_s}{\partial \xi_{\pm}^2},
\end{equation}
where
\begin{eqnarray}
\label{expressions-F-pm}
f^{\pm}_s(\xi_{\pm},T) & = & - \sum_{n \in \mathbb{Z}\backslash\{0\}} \sum_{k \in \mathbb{Z} \backslash \{0,n\}}
\frac{(n-k)^2}{4 n k} |a^{\mp}_{k}(T)|^2 a^{\pm}_n(T)  e^{\frac{i \pi n \xi_{\pm}}{L}} \\
\nonumber
& = & \frac{1}{2} f^{\pm}(\xi_{\pm},T) \left(  \sum_{k \in \mathbb{Z} \backslash \{0\}} |a^{\mp}_{k}(T)|^2 \right) \\
\nonumber
& = & \frac{1}{4L} f^{\pm}(\xi_{\pm},T) \int_{-L}^L |f^{\mp}(\xi,T)|^2 d \xi.
\end{eqnarray}
The initial-value problem for the linearized KdV
equations (\ref{KdV-lin}) starts with the initial values $\phi^{\pm} |_{T = 0}$
given by (\ref{phi-pm-Fourier}) and (\ref{ic-kdv-lin})
(the closed form expressions in terms of the leading order solutions $f^{\pm}$
can be found in \cite{Karima1,Karima2}).
There exists a unique global solution $\phi^{\pm} \in C(\mathbb{R}_+,H^s_{\rm per}(-L,L))$
for any $\phi^{\pm} |_{T = 0} \in H^s_{\rm per}(-L,L)$
with $s \geq -\frac{1}{2}$ provided that the source term on the right-hand-side of
(\ref{KdV-lin}) is sufficiently smooth in $T$ and $\xi_{\pm}$.

After the constraints (\ref{constraints-second-order}) are substituted back into the differential equations
(\ref{second-order-ode}), we can obtain a bounded solution for $U_2(x,t)$. This completes
the construction of the formal asymptotic
expansion (\ref{asymptotic-expansion}) up to and including the $\mathcal{O}(\epsilon^2)$ terms.

\subsection{Justification}

We shall now justify the long-wave approximation. The following theorem gives
the main result of the justification analysis.

\begin{theorem}
Assume that $(F,V) \in H^1_{\rm per}(-L,L) \times H^1_{\rm per}(-L,L)$
subject to the zero-mean constraint (\ref{velocity-constraint}) on $V$.
Fix $s \geq 10$ and let $f^{\pm} \in C(\mathbb{R},H^s_{\rm per}(-L,L))$ be global solutions of the KdV equations
(\ref{KdV}) starting with the initial data (\ref{leading-order-before-ic}).
Let $U_0$ and $U_1$ be given by (\ref{leading-order}) and (\ref{first-order-correction})
with (\ref{f-c-Fourier}), (\ref{phi-pm-Fourier}), and (\ref{ic-kdv-lin}). There is $\epsilon_0 > 0$
such that for all $\epsilon \in (0,\epsilon_0)$ and all $\epsilon$-independent
$T_0 > 0$, there is an $\epsilon$-independent constant $C > 0$ such that
for all $t_0 \in [0,T_0/\epsilon]$, the local solution of the
regularized Boussinesq equation (\ref{B}) satisfies
\begin{equation}
\label{bound-1}
\sup_{t \in [0,t_0]} \| U - c_0 - U_0 - \epsilon U_1 \|_{H^1_{\rm per}} \leq C \epsilon^2 t_0.
\end{equation}
If, in addition, $\phi^{\pm}$ in (\ref{phi-pm-Fourier}) satisfies the linearized KdV equations
(\ref{KdV-lin}) subject to the initial data (\ref{ic-kdv-lin}) and $s$ is sufficiently large,
then for all $\epsilon \in (0,\epsilon_0)$ and all $\epsilon$-independent
$T_0 > 0$, there is an $\epsilon$-independent constant $C > 0$ such that
\begin{equation}
\label{bound-2}
\sup_{t \in [0,T_0/\epsilon]} \| U - c_0 - U_0 - \epsilon U_1 \|_{H^1_{\rm per}} \leq C \epsilon^2.
\end{equation}
\label{theorem-main}
\end{theorem}

Before proving the theorem, we make some relevant remarks.

\begin{remark}
The bound (\ref{bound-1}) of Theorem \ref{theorem-main} generalizes the result of Theorem 7 in \cite{Bona}
obtained in the context of Boussinesq systems. However, if the authors of \cite{Bona}
restrict their consideration to the zero mean value for the initial data (case (ii') in Theorem 7)
and show that the nonzero mean values do not produce good long-wave approximations
(cases (ii) and (iii) in Theorem 7), we show that the mean value in the initial data for $U|_{t = 0} = F$
can be naturally incorporated in the justification analysis by modifying the
velocity term of the uncoupled KdV equations (\ref{KdV}).
\end{remark}

\begin{remark}
The difference between bounds (\ref{bound-1}) and (\ref{bound-2}) of Theorem \ref{theorem-main}
is in the time scales, for which the first-order correction terms remain valid.
Bound (\ref{bound-1}) shows that the error of the long-wave approximation
is of the $\mathcal{O}(\epsilon^2)$ order at the time scale $t_0 = \mathcal{O}(1)$
but becomes comparable with the $\mathcal{O}(\epsilon)$ first-order correction terms
at the time scale $t_0 = \mathcal{O}(\epsilon^{-1})$. On the other hand, bound (\ref{bound-2}) shows that the
error of the long-wave approximation remains of the $\mathcal{O}(\epsilon^2)$ order
at the time scale $t_0 = \mathcal{O}(\epsilon^{-1})$ if the first-order correction terms
satisfy the linearized KdV equations (\ref{KdV-lin}). This improved result corresponds to Theorem 1.1
in \cite{WW} on the infinite line with the only difference that the justification analysis is performed in the energy space
of the regularized Boussinesq equation (\ref{B}) compared to the space $H^{\sigma} \cap H^{\sigma + 8}$ with $\sigma \geq 4$
used in \cite{WW}.
\end{remark}

\begin{remark}
Figures 3(c) and 5(c) in \cite{Karima2} illustrate the growth of the approximation error
without the account of the linearized KdV equation (\ref{KdV-lin}) at the time scale
$t_0 = \mathcal{O}(\epsilon^{-1})$. Despite the fact that
the first-order correction terms were found to give a smaller approximation error,
the comparable $\mathcal{O}(\epsilon)$ behavior between the first-order
correction terms and the approximation errors was observed on these figures
at the time scale $t_0 = \mathcal{O}(\epsilon^{-1})$.
\end{remark}

\begin{proof1}{\em of Theorem \ref{theorem-main}.}
We shall first prove bound (\ref{bound-1}) on the approximation error.
We work in energy space of the regularized Boussinesq equation (\ref{B})
and write the approximation in the form
\begin{equation}
\label{decomposition-B}
U(x,t) = c_0 + U_0(x,t) + \epsilon U_1(x,t) + \epsilon^2 \hat{U}(x,t),
\end{equation}
where $c_0 = F_0$ is the mean value, $U_0$ and $U_1$ are given explicitly
by (\ref{leading-order}) and (\ref{first-order-correction})
with (\ref{f-c-Fourier}), (\ref{phi-pm-Fourier}), and (\ref{ic-kdv-lin}), $f^{\pm}$
satisfy the uncoupled KdV equations (\ref{KdV}), and $\hat{U}$ is the
error term that depends on $\epsilon$. Substituting the decomposition (\ref{decomposition-B})
into the regularized Boussinesq equation (\ref{B}), we obtain the time evolution
problem for the error term:
\begin{eqnarray}
\label{B-error}
\hat{U}_{tt} - (1 + \epsilon c_0) \hat{U}_{xx} - \epsilon \hat{U}_{ttxx} =
\epsilon \partial_x^2 \left( U_0 \hat{U} + \epsilon U_1 \hat{U} + \frac{1}{2} \epsilon^2 \hat{U}^2  \right) +
\hat{H},
\end{eqnarray}
where the initial data are
\begin{equation}
\label{initial-data-B-error}
\hat{U} |_{t = 0} = 0, \quad \hat{U}_t |_{t = 0} = -\partial_T U_1 |_{T = 0},
\end{equation}
and the source term is
\begin{eqnarray*}
\hat{H} & = & - 2 \partial_t \partial_T U_1 - \partial_T^2 U_0 - \epsilon \partial_T^2 U_1
+ c_0 \partial_x^2 U_1 + (\partial_t + \epsilon \partial_T)^2 \partial_x^2 U_1
+ 2 \partial_t \partial_T \partial_x^2 U_0 + \epsilon \partial_T^2 \partial_x^2 U_0  \\
& \phantom{t} & + \partial_x^2 (U_0 U_1) + \frac{1}{2} \epsilon \partial_x^2 (U_1^2).
\end{eqnarray*}

We use a priori energy estimates (see, e.g., \cite{PelPon,PelSch} for similar applications
of this technique). By an extension of Proposition \ref{proposition-local},
there exists a unique solution
$$
\hat{U} \in C^1([0,t_0],H^1_{\rm per}(-L,L))
$$
of the perturbed regularized Boussinesq equation (\ref{B-error})
for some $\epsilon$-independent $t_0 > 0$ starting with the initial data (\ref{initial-data-B-error})
provided that the source term satisfies
\begin{equation}
\label{requirement-H}
\hat{H} \in C([0,t_0],H^1_{\rm per}(-L,L)).
\end{equation}
By looking at the explicit expression for $\hat{H}$, where
$U_0$ and $U_1$ are given by (\ref{leading-order}) and (\ref{first-order-correction}) and
$f^{\pm} \in C(\mathbb{R},H^s_{\rm per}(-L,L))$ are global solutions of the uncoupled KdV equations
(\ref{KdV}), we realizes that the term of the highest regularity is $\partial_T^2 \partial_x^2 U_1 \sim
\partial_{\xi_{\pm}}^9 f^{\pm}$, hence $\hat{H}$ is from the class (\ref{requirement-H}) if $s \geq 10$.

Let us introduce the energy at the local solution $\hat{U} \in C^1([0,t_0],H^1_{\rm per}(-L,L))$
of the perturbed regularized Boussinesq equation (\ref{B-error}):
\begin{equation}
\label{energy-norm}
\hat{E} = \int_{-L}^L \left( \hat{U}_t^2 + (1 + \epsilon c_0) \hat{U}_x^2 +
\epsilon \hat{U}_{tx}^2 + \epsilon U_0 \hat{U}_x^2
+ 2 \epsilon U_{0x} \hat{U} \hat{U}_x \right) dx.
\end{equation}
Multiplying (\ref{B-error}) by $\hat{U}_t$, integrating in $x$ over $[-L,L]$,
and using approximation arguments in Sobolev spaces of higher regularity, we obtain
$$
\frac{1}{2} \frac{d \hat{E}}{dt} = \int_{-L}^L \left[ \hat{H} \hat{U}_t
+ \epsilon \left( U_{0xt} \hat{U} \hat{U}_x + U_{0x} \hat{U}_t \hat{U}_x +
\frac{1}{2} U_{0t} \hat{U}^2_x \right)
- \epsilon^2 \hat{U}_{tx} \left( U_1 \hat{U} + \frac{1}{2} \epsilon \hat{U}^2  \right)_x \right] dx.
$$

Because $x$ and $t$ derivatives of $U_0$ are $\epsilon$-independent,
for sufficiently small $\epsilon$, there is an $\epsilon$-independent positive constant $C$
such that
$$
\| \hat{U}_t \|^2_{L^2_{\rm per}} + \| \hat{U}_x \|^2_{L^2_{\rm per}} + \epsilon \| \hat{U}_{xt} \|^2_{L^2_{\rm per}} \leq C \hat{E}.
$$
By Poincar\'e's inequality for $(2L)$-periodic mean-zero functions, there is
another positive constant $C$ such that
$$
\| \hat{U} \|^2_{L^2_{\rm per}} \leq C \| \hat{U}_x \|^2_{L^2_{\rm per}} \leq C \hat{E}.
$$
By Cauchy--Schwarz's inequality, we obtain from the energy balance equation that
\begin{equation}
\label{energy-inequality}
\frac{1}{2} \frac{d \hat{E}}{dt} \leq \| \hat{H} \|_{L^2_{\rm per}}  \| \hat{U}_t \|_{L^2_{\rm per}}
+ C \epsilon \left( \| U_0 \|_{L^{\infty}_{\rm per}} + \epsilon^{1/2} \| U_1 \|_{L^{\infty}_{\rm per}}
+ \epsilon^{3/2} \| \hat{U} \|_{L^{\infty}_{\rm per}} \right) \hat{E},
\end{equation}
where the positive constant $C$ is $\epsilon$-independent. Note that
the terms $\| U_{0t} \|_{L^{\infty}_{\rm per}}$, $\| U_{0x} \|_{L^{\infty}_{\rm per}}$,
$\| U_{0xt} \|_{L^{\infty}_{\rm per}}$, and $\| U_{1x} \|_{L^{\infty}_{\rm per}}$
are not listed in the inequality (\ref{energy-inequality}) because they are $\epsilon$-independent
and bounded if $f^{\pm} \in C(\mathbb{R},H^s_{\rm per}(-L,L))$ with $s \geq 10$.

Setting $\hat{E} := \hat{Q}^2$ and using Sobolev embedding's
$\| \hat{U} \|_{L^{\infty}_{\rm per}} \leq C_{\rm emb} \hat{Q}$,
we rewrite a priori energy estimate (\ref{energy-inequality}) in the form
\begin{equation}
\label{energy-estimate}
\frac{d \hat{Q}}{dt} \leq \| \hat{H} \|_{L^2_{\rm per}} +
C \epsilon \left( \| U_0 \|_{L^{\infty}_{\rm per}} + \epsilon^{1/2} \| U_1 \|_{L^{\infty}_{\rm per}}
+ \epsilon^{3/2} \hat{Q} \right)   \hat{Q},
\end{equation}
for another positive $\epsilon$-independent constant $C$. By Gronwall's inequality,
we integrate the a priori energy estimate (\ref{energy-estimate}) to obtain
\begin{equation}
\label{bound-tech-1}
\hat{Q}(t) \leq \left( \hat{Q}(0) + t_0 \sup_{t \in [0,t_0]} \| \hat{H} \|_{L^2_{\rm per}} \right)
e^{C \epsilon t}, \quad t \in [0,t_0],
\end{equation}
for any $t_0 > 0$, sufficiently small $\epsilon$, and some $(t_0,\epsilon)$-independent positive
constant $C_0$. Since $\hat{Q}(0) = \| \partial_T U_1 \|_{L^2_{\rm per}}$,
bound (\ref{bound-tech-1}) yields the result (\ref{bound-1}) after returning
to the original variables (\ref{decomposition-B}). By Proposition \ref{proposition-continuation},
the solution is continued from $t_0 = \mathcal{O}(1)$ to $t_0 = \mathcal{O}(\epsilon^{-1})$
thanks to $\hat{Q}(t) < \infty$ for all $t \in [0,t_0]$ and
all sufficiently small $\epsilon > 0$, as well as to Sobolev's embedding
$$
\| \epsilon^2 \hat{U} \|_{L^{\infty}_{\rm per}} \leq C_{\rm emb} \epsilon^2 \hat{Q}, \quad
\| \epsilon^2 \hat{U}_t \|_{L^{\infty}_{\rm per}} \leq C_{\rm emb} \epsilon^{3/2} \hat{Q}.
$$

Bound (\ref{bound-2}) is proved similarly after writing $\hat{U} = U_2 + \epsilon \check{U}$,
where $U_2$ is a bounded solution of system (\ref{second-order}), whereas $\phi^{\pm}$ in
(\ref{phi-pm-Fourier}) satisfy the linearized KdV equations
(\ref{KdV-lin}) subject to the initial data (\ref{ic-kdv-lin}). The new error
term $\check{U}$ satisfy a priori energy estimate similar to (\ref{energy-estimate})
with the new error term. This energy estimate yields bound (\ref{bound-2})
 after returning to the original variables (\ref{decomposition-B})
 thanks to the triangle inequality and the bound
$\| U_2 \|_{H^1_{\rm per}} \leq C$, where the positive constant $C$ is independent of
$t_0 = \mathcal{O}(\epsilon^{-1})$ and $\epsilon$.
\end{proof1}

\section{Long-wave approximation for $\gamma > 0$}

We shall consider here the case of the regularized Boussinesq--Ostrovsky equation
(\ref{BO}) with $\gamma > 0$. We consider again the initial-value problem
starting with the initial data (\ref{initial-values}) and (\ref{FS-IO})
satisfying the zero-mean velocity constraint (\ref{velocity-constraint}).
By the exact solution (\ref{mean-value}) the mean value of the solution is oscillating in $t$
with $\langle U \rangle = F_0 \cos(\sqrt{\epsilon \gamma} t)$.

Substituting $U(x,t) = c_0 \cos(\omega t) + \tilde{U}(x,t)$ into the evolution equation
(\ref{BO}), where $c_0 := F_0$, $\omega := \sqrt{\epsilon \gamma}$, and
$\tilde{U}$ is the zero-mean part of the $2L$-periodic function $U$,
we obtain the evolution equation
\begin{equation}
\label{BO-periodic-O}
\tilde{U}_{tt} - \tilde{U}_{xx} = \epsilon \left( c_0 \cos(\omega t) \tilde{U}_{xx} +
\frac{1}{2} (\tilde{U}^2)_{xx} + \tilde{U}_{tt xx} - \gamma \tilde{U} \right).
\end{equation}
By Proposition \ref{proposition-local}, there exists a unique local solution
$\tilde{U} \in C^1([0,t_0],H^s_{\rm per}(-L,L))$ of the evolution equation (\ref{BO-periodic-O})
for any $(\tilde{F},\tilde{V}) \in H^s_{\rm per}(-L,L) \times H^{s}_{\rm per}(-L,L)$
with $s > \frac{1}{2}$, where $t_0 > 0$ is a local existence time.

\subsection{Derivation}

We shall repeat steps of the formal asymptotic theory, which relies on the
decomposition (\ref{asymptotic-expansion}) and the leading-order approximation
(\ref{leading-order}), with the initial conditions (\ref{ic-kdv}). In what follows,
we first work implicitly with $\epsilon$-dependent $\omega$ and then estimate
the size of the correction terms by using the explicit
dependence $\omega = \sqrt{\epsilon \gamma}$.

\vspace{0.2cm}

{\bf Order $\mathcal{O}(\epsilon)$:} The first-order correction term
satisfies the initial-value problem:
\begin{equation}
\label{first-order-O}
 \left\{ \begin{array}{l} (\partial_t^2 - \partial_x^2) U_1 =
-2 \partial^2_{t T} U_0 + c_0 \cos(\omega t) \partial^2_{x} U_0
+ \frac{1}{2} \partial_x^2(U_0^2) + \partial^4_{ttxx} U_0 - \gamma U_0, \\
U_1 |_{t = 0} = 0, \\
\partial_t U_1 |_{t = 0} = - \partial_T U_0 |_{t = 0}. \end{array} \right.
\end{equation}
Using the Fourier series $U_1(x,t) = \sum_{n \in \mathbb{Z} \backslash \{0\}}
g_n(t) e^{\frac{i \pi n x}{L}}$, we obtain the uncoupled system of differential equations
\begin{equation}
\label{first-order-ode-BO}
\frac{d^2 g_n}{d t^2} + \left( \frac{\pi n}{L} \right)^2 g_n = h_n(t),
\end{equation}
where
\begin{eqnarray*}
h_n(t) & = & -\frac{2 \pi i n}{L} \left( \frac{d a^+_n}{d T} e^{\frac{i \pi n t}{L}} - \frac{d a^-_n}{d T}
e^{-\frac{i \pi n t}{L}} \right) + \frac{\pi^4 n^4}{L^4} \left( a^+_n e^{\frac{i \pi n t}{L}} + a^-_n
e^{-\frac{i \pi n t}{L}} \right) \\
& \phantom{t} & - \frac{\pi^2 n^2}{L^2} c_0 \cos(\omega t) \left( a^+_n e^{\frac{i \pi n t}{L}} + a^-_n
e^{-\frac{i \pi n t}{L}} \right) - \gamma  \left( a^+_n e^{\frac{i \pi n t}{L}} + a^-_n e^{-\frac{i \pi n t}{L}} \right)\\
& \phantom{t} & -\frac{\pi^2 n^2}{2 L^2}
\left( \sum_{k \in \mathbb{Z} \backslash \{0,n\}} a^+_k a^+_{n-k} \right) e^{\frac{i \pi n t}{L}}
-\frac{\pi^2 n^2}{2 L^2} \left( \sum_{k \in \mathbb{Z}\backslash \{0,n\}} a^-_k a^-_{n-k} \right) e^{-\frac{i \pi n t}{L}}  \\
& \phantom{t} & -\frac{\pi^2 n^2}{L^2} \left( \sum_{k \in \mathbb{Z}\backslash \{0,n\}} a^+_k a^-_{n-k} e^{\frac{i \pi (2k-n) t}{L}} \right).
\end{eqnarray*}

The resonant terms at $e^{\pm \frac{i \pi n t}{L}}$ are removed if the Fourier coefficients
satisfy the evolution equations:
\begin{eqnarray}
\label{constraints-BO}
\mp \frac{2 \pi i n}{L} \frac{d a^{\pm}_n}{d T}
+ \frac{\pi^4 n^4}{L^4} a^{\pm}_n - \gamma  a^{\pm}_n
- \frac{\pi^2 n^2}{2 L^2}
\sum_{k \in \mathbb{Z} \backslash \{0,n\}} a^{\pm}_k a^{\pm}_{n-k} = 0,
\end{eqnarray}
which are equivalent to the two uncoupled Ostrovsky equations
\begin{equation}
\label{KdV-Ost}
\frac{\partial}{\partial \xi_{\pm}} \left( \mp 2 \frac{\partial f^{\pm}}{\partial T}
+ \frac{\partial^3 f^{\pm}}{\partial \xi_{\pm}^3}
+ f^{\pm} \frac{\partial f^{\pm}}{\partial \xi_{\pm}} \right) = \gamma f^{\pm}.
\end{equation}
We consider the initial-value problem for the uncoupled Ostrovsky
equations (\ref{KdV-Ost}) starting with the initial values $f^{\pm} |_{T = 0}$
given by (\ref{leading-order-before-ic}).
By the local and global well-posedness theory for the Ostrovsky equation
\cite{GL07,LM06,Ts,VL04}, a unique global solution
$f^{\pm} \in C(\mathbb{R}_+,H^s_{\rm per}(-L,L))$ exists
for any $f^{\pm} |_{T = 0} \in H^s_{\rm per}(-L,L)$ with $s > \frac{3}{4}$.

\begin{remark}
If $\gamma \neq 0$, solutions of the Ostrovsky equations
(\ref{KdV-Ost}) must satisfy the zero-mean constraints \cite{Benilov}. In our derivation,
the zero-mean constraints are satisfied automatically because $U_0$ and $f^{\pm}$
are the zero-mean parts of the $(2L)$-periodic functions. Note that the oscillating term
$\langle U \rangle = c_0 \cos(\omega t)$ does not contribute
to the Ostrovsky equations (\ref{KdV-Ost}). Consequently, as $\gamma \to 0$, the
limiting KdV equation (\ref{KdV-Ost}) is different from the KdV equation (\ref{KdV})
if $c_0 \neq 0$.
\end{remark}

After the constraints (\ref{constraints-BO}) are substituted back into the differential equations
(\ref{first-order-ode-BO}), we obtain the linear inhomogeneous equations
\begin{eqnarray*}
\frac{d^2 g_n}{d t^2} + \left( \frac{\pi n}{L} \right)^2 g_n & = &
- \frac{\pi^2 n^2}{L^2} c_0 \cos(\omega t) \left( a^+_n e^{\frac{i \pi n t}{L}} + a^-_n
e^{-\frac{i \pi n t}{L}} \right) \\
& \phantom{t} & \phantom{texttexttext}
-\frac{\pi^2 n^2}{L^2} \sum_{k \in \mathbb{Z}\backslash\{0,n\}} a^+_k a^-_{n-k} e^{\frac{i \pi (2k-n) t}{L}},
%\label{first-order-ode-ivc-BO}
\end{eqnarray*}
subject to the initial conditions
\begin{equation*}
g_n(0) = 0, \quad \partial_t g_n(0) = - \partial_T a_n^+(0) - \partial_T a_n^-(0).
\end{equation*}
This initial-value problem admits the following bounded solution:
\begin{eqnarray*}
g_n(t) & = & \frac{c_0 \pi^2 n^2}{4 \pi^2 n^2 - \omega^2 L^2}
\left[\frac{2\pi n i}{L} \frac{\sin(\omega t)}{\omega}
( a_n^+ e^{\frac{i \pi n t}{L}} - a_n^- e^{\frac{-i \pi n t}{L}})
- \cos(\omega t) ( a_n^+ e^{\frac{i \pi n t}{L}} + a_n^- e^{\frac{-i \pi n t}{L}}) \right] \\
& \phantom{t} & + \sum_{k \in \mathbb{Z} \backslash \{0,n\}} \frac{n^2}{4 k (k-n)} a_k^+ a_{n-k}^-
e^{\frac{i \pi (2k-n) t}{L}} + G_n \cos\left(\frac{\pi n t}{L}\right) + H_n  \sin\left(\frac{\pi n t}{L}\right),
%\label{expl-sol-1-BO}
\end{eqnarray*}
where $G_n$ and $H_n$ are constants of integrations to be found from the initial conditions for $g_n$.

\begin{remark}
Since $\omega = \sqrt{\epsilon \gamma}$, the first term in the explicit solution
for $g_n$ grows in $t$ in the limit $\epsilon \to 0$ if $c_0 \neq 0$ but it is
nevertheless bounded by the $\mathcal{O}(\epsilon^{-1/2})$ constant for any $\epsilon > 0$. This
fact implies that
$$
\| \epsilon U_1 \|_{H^1_{\rm per}} = \mathcal{O}(c_0 \epsilon^{1/2}) \quad \mbox{\rm as} \quad \epsilon \to 0.
$$
If the zero-mean velocity constraint (\ref{velocity-constraint}) is violated and $V_0 \neq 0$,
then the asymptotic procedure will give $\| \epsilon U_1 \|_{H^1_{\rm per}} = \mathcal{O}(1)$ as $\epsilon \to 0$
and the first-order correction term (as well as all higher-order correction terms) become comparable with
the leading-order approximation. This will clearly prevent us from justification of the long-wave approximation.
This remark explains why we have set $V_0 = 0$ for $\gamma > 0$ in the constraint (\ref{velocity-constraint}).
\end{remark}

Using the explicit solution for $g_n$, we rewrite
the first-order correction term in the implicit form:
\begin{equation}
\label{first-order-correction-BO}
U_1(x,t) = f_c(x,t) + f_m(x,t) + \phi^-(\xi_-,T) + \phi^+(\xi_+,T),
\end{equation}
where $f_c$ is given by (\ref{f-c-Fourier}) and (\ref{expression}), $f_m$ is uniquely defined by
\begin{eqnarray}
\nonumber f_m(x,t) & = & \frac{\sin(\omega t)}{\omega}
\sum_{n \in \mathbb{Z} \backslash \{0\}} \frac{c_0 \pi^2 n^2}{4 \pi^2 n^2 - \omega^2 L^2}
\frac{2\pi n i}{L}
( a_n^+(T) e^{\frac{i \pi n t}{L} + \frac{i \pi n x}{L}} - a_n^-(T) e^{\frac{-i \pi n t}{L} + \frac{i \pi n x}{L}}) \\
& \phantom{t} &
- \cos(\omega t) \sum_{n \in \mathbb{Z} \backslash \{0\}} \frac{c_0 \pi^2 n^2}{4 \pi^2 n^2 - \omega^2 L^2}
( a_n^+(T) e^{\frac{i \pi n t}{L}+ \frac{i \pi n x}{L}} + a_n^-(T) e^{\frac{-i \pi n t}{L}+ \frac{i \pi n x}{L}}).
\label{expression-f-c} 
\end{eqnarray}
and functions $\phi^{\pm}$ are given by
\begin{equation}
\label{first-order-correction-BO-phi}
\phi^{\pm}(\xi_{\pm},T) = \sum_{n \in \mathbb{Z}} b^{\pm}_n(T) e^{\frac{i \pi n \xi_{\pm}}{L}},
\end{equation}
subject to the initial conditions
\begin{eqnarray}
\nonumber
b^{\pm}_n |_{T = 0} & = & \frac{c_0 \pi^2 n^2}{4 \pi^2 n^2 - \omega^2 L^2} a_n^{\mp} |_{T = 0}
- \sum_{k \in \mathbb{Z} \backslash \{0,n\}}
\frac{n (n \pm (2k-n))}{8 k (k-n)} \left( a_k^+ a_{n-k}^- \right)|_{T = 0} \\
& \phantom{t} &
\mp \frac{L}{2 i \pi n} \left( \frac{d a_n^+}{d T} + \frac{d a_n^-}{d T}\right) \biggr|_{T = 0}.
\label{ic-kdv-lin-BO}
\end{eqnarray}

%Using elementary computations, we can represent the correction term $f_m$ in the equivalent form
%\begin{eqnarray*}
%f_m(x,t) & = & -\frac{c_0}{2} \left( f^+(\xi_+) + f^-(\xi_-) \right) \cos\left(\frac{\omega}{2}(\xi_+ - \xi_-) \right) \\
%& \phantom{t} & +\frac{c_0}{2 \omega} \left( \partial_{\xi_+} - \frac{\omega^2}{4} \partial_{\xi_+}^{-1} \right)
%f^+(\xi_+) \sin\left(\frac{\omega}{2}(\xi_+ - \xi_-) \right)\\
%& \phantom{t} & -\frac{c_0}{2 \omega} \left( \partial_{\xi_-} - \frac{\omega^2}{4} \partial_{\xi_-}^{-1} \right)
%f^-(\xi_-) \sin\left(\frac{\omega}{2}(\xi_+ - \xi_-) \right),
%\end{eqnarray*}
%where we can recall that $\xi_+ - \xi_- = 2t$.

\vspace{0.2cm}

{\bf Order $\mathcal{O}(\epsilon^2)$:} The second-order correction term
satisfies the initial-value problem:
\begin{equation}
\label{second-order-BO}
\left\{ \begin{array}{l} (\partial_t^2 - \partial_x^2) U_2 =
-2 \partial^2_{t T} U_1 - \partial_T^2 U_0 + c_0 \cos(\omega t) \partial_x^2 U_1 + \partial_x^2(U_0 U_1) \\
\phantom{texttexttexttexttexttext} +
\partial^4_{ttxx} U_1 + 2\partial^4_{tTxx} U_0  - \gamma U_1, \\
U_2 |_{t = 0} = 0, \\
\partial_t U_2 |_{t = 0} = - \partial_T U_1 |_{t = 0}. \end{array} \right.
\end{equation}

Repeating the procedure of removing the resonant terms $e^{\pm \frac{i \pi n t}{L}}$
and using again the fact that $\cos(\omega t)$ and $\sin(\omega t)$
do not produce the resonant terms, we define uniquely the evolution problem for the Fourier coefficients
of the functions $\phi^{\pm}$ in  (\ref{first-order-correction-BO-phi}):
\begin{eqnarray}
\label{constraints-second-order-BO}
- \frac{d^2 a_n^{\pm}}{d T^2} \mp \frac{2 \pi i n}{L} \frac{d b^{\pm}_n}{d T}
+ \frac{\pi^4 n^4}{L^4} b^{\pm}_n \mp \frac{2i \pi^3 n^3}{L^3} \frac{d a_n^{\pm}}{d T} \\
- \frac{\pi^2 n^2}{2 L^2} \sum_{k \in \mathbb{Z} \backslash \{0,n\}} a^{\pm}_k b^{\pm}_{n-k}
- \frac{\pi^2 n^2}{L^2} \sum_{k \in \mathbb{Z} \backslash \{0,n\}} \frac{(n-k)^2}{4 n k} a^{\pm}_n |a^{\mp}_{k}|^2
- \gamma  b^{\pm}_n = 0.
\nonumber
\end{eqnarray}
These equations are equivalent to the linearized Ostrovsky equations
\begin{equation}
\label{KdV-Ost-lin}
\frac{\partial}{\partial \xi_{\pm}} \left( \mp 2 \frac{\partial \phi^{\pm}}{\partial T}
+ \frac{\partial^3 \phi^{\pm}}{\partial \xi_{\pm}^3}
+ \frac{\partial}{\partial \xi_{\pm}} f^{\pm} \phi^{\pm} \right) = \gamma \phi^{\pm}
+ \frac{\partial^2 f^{\pm}}{\partial T^2} \mp 2 \frac{\partial^4 f^{\pm}}{\partial \xi_{\pm}^3 T}
+ \frac{\partial^2 f_s^{\pm}}{\partial \xi_{\pm}^2},
\end{equation}
with the same definition for $f_s^{\pm}$ as in (\ref{expressions-F-pm}). The closed form expressions for the initial conditions in terms of the leading order solutions $f^{\pm}$
can be found in \cite{Karima1,Karima2}.

After the constraints (\ref{constraints-second-order-BO}) are substituted back to the initial-value problem
(\ref{second-order-BO}), we can obtain a bounded solution for $U_2(x,t)$. Note that
the bounded solutions satisfies
$$
\| \epsilon^2 U_2 \|_{H^1_{\rm per}} = \mathcal{O}(c_0 \epsilon) \quad \mbox{\rm as} \quad \epsilon \to 0
$$
because of the oscillatory behavior of the functions $\cos(\omega t)$ and $\sin(\omega t)$
with $\omega = \sqrt{\epsilon \gamma}$.
This completes the construction of the formal asymptotic
expansion (\ref{asymptotic-expansion}) up to and including
the $\mathcal{O}(\epsilon^2)$ terms.

\subsection{Justification}

The following theorem gives the main result on the justification of the long-wave approximation.

\begin{theorem}
Assume that $(F,V) \in H^1_{\rm per}(-L,L) \times H^1_{\rm per}(-L,L)$
subject to the zero-mean constraint (\ref{velocity-constraint}) on $V$.
Fix $s \geq 10$ and let $f^{\pm} \in C(\mathbb{R},H^s_{\rm per}(-L,L))$ be global solutions of the Ostrovsky equations
(\ref{KdV-Ost}) starting with the initial conditions (\ref{leading-order-before-ic}).
Let $U_0$ and $U_1$ be given by (\ref{leading-order}) and (\ref{first-order-correction-BO})
with (\ref{expression}), (\ref{expression-f-c}), (\ref{first-order-correction-BO-phi}),
and (\ref{ic-kdv-lin-BO}). There is $\epsilon_0 > 0$
such that for all $\epsilon \in (0,\epsilon_0)$ and all $\epsilon$-independent
$T_0 > 0$, there is $\epsilon$-independent constant $C > 0$
such that for all $t_0 \in [0,T_0/\epsilon]$, the local solution of
the regularized Boussinesq--Ostrovsky equation (\ref{BO}) satisfies
\begin{equation}
\label{bound-1-BO}
\sup_{t \in [0,t_0]} \| U - c_0 \cos(\omega t) - U_0 - \epsilon U_1 \|_{H^1_{\rm per}} \leq C \epsilon t_0 (c_0 + \epsilon).
\end{equation}
If, in addition, $\phi^{\pm}$ in (\ref{first-order-correction-BO-phi}),  satisfies the linearized Ostrovsky equations
(\ref{KdV-Ost-lin}) subject to the initial conditions (\ref{ic-kdv-lin-BO}) and $s$ is sufficiently large,
then for all $\epsilon \in (0,\epsilon_0)$ and all $\epsilon$-independent $T_0 > 0$, there is an $\epsilon$-independent
constant $C > 0$ such that
\begin{equation}
\label{bound-2-BO}
\sup_{t \in [0,T_0/\epsilon]} \| U - c_0 \cos(\omega t)  - U_0 - \epsilon U_1 \|_{H^1_{\rm per}} \leq
C \epsilon (c_0 + \epsilon).
\end{equation}
\label{theorem-main-2}
\end{theorem}

\begin{remark}
The bounds (\ref{bound-1-BO}) and (\ref{bound-2-BO}) of Theorem \ref{theorem-main-2} are larger
then the bounds (\ref{bound-1}) and (\ref{bound-2}) of Theorem \ref{theorem-main} if $c_0 \neq 0$ but they still
complete justification of the long-wave approximation up to
the first-order correction term because
$\| \epsilon U_1 \|_{H^1_{\rm per}} = \mathcal{O}(c_0 \epsilon^{1/2})$ as $\epsilon \to 0$.
\end{remark}

\begin{remark}
If $c_0 = 0$, then bounds (\ref{bound-1-BO}) and (\ref{bound-2-BO}) become
bounds (\ref{bound-1}) and (\ref{bound-2}) because the periodic driving terms
$\cos(\omega t)$ and $\sin(\omega t)$ are not present in all the expansions.
\end{remark}

The proof of this justification theorem is similar to the proof of Theorem \ref{theorem-main}
after the energy of the error term for the formal asymptotic expansion as in
(\ref{decomposition-B}) and (\ref{energy-norm}) is modified to include
the $\epsilon \gamma$-term as in the energy function (\ref{energy-function}).

\section{Numerical illustrations}

\subsection{Boussinesq equation}

Let us consider the initial-value problem for the regularized Boussinesq equation (\ref{B})
starting with initial data
\begin{eqnarray}
\left\{ \begin{array}{l} U|_{t=0} = 3k^2 \ {\rm sech}^2 \left(\frac{kx}{2}\right), \\
U_t|_{t=0} = 3k^3 \ {\rm sech}^2 \left(\frac{kx}{2}\right)  \ {\rm tanh} \left(\frac{kx}{2}\right),
\end{array} \right.
\label{BO_f+0}
\end{eqnarray}
where $k > 0$ is an arbitrary parameter. The initial data is defined on the periodic domain 
$-L \leq x \leq L$ and the mean value is given by
$$
c_0 =\frac{1}{2L} \int^{L}_{-L} U|_{t=0} dx = \frac{6k}{L}\ {\rm tanh} \left(\frac{kL}{2}\right).
$$
When $L \to \infty$, $c_0 \to 0$, and the initial data (\ref{BO_f+0}) corresponds at
the leading order to a solitary wave of the KdV equation (\ref{KdV}) propagating to the right:
\begin{eqnarray}
f^+(\xi_+,T) = 0, \quad
f^-(\xi_-,T) = 3 k^2 {\rm sech}^2(z), \quad
z = \frac{k}{2} \left(\xi_{-} - \frac{k^2}{2} T \right).
\end{eqnarray}

Combining (\ref{leading-order}) and (\ref{first-order-correction}), we obtain
the weakly nonlinear solution in the form
\begin{equation}
\label{weakly-nonlinear-solution}
U = f^-(\xi_-,T)  + \epsilon \left[ \phi^-(\xi_-,T) + \phi^+(\xi_+,T) \right]
+ \mathcal{O}(\varepsilon^2) \ ,
\end{equation}
where the correction terms $\phi^{\pm}$ satisfy the linearized KdV equations
\begin{eqnarray}
-2 \partial_T \phi^+ + \partial_{\xi_+}^3 \phi^+ = 0
\label{phi_plus}
\end{eqnarray}
and
\begin{eqnarray}
\partial_{\xi^-} \left( 2 \partial_T \phi^- +
\partial_{\xi_-}^3 \phi^- + \partial_{\xi_-} (f^- \phi^-) \right)
= \partial_T^2 f^- + 2 \partial_T \partial_{\xi_-}^3 f^-,
\label{phi_minus}
\end{eqnarray}
subject to the initial data
\begin{equation}
\label{phi-initial-data}
\phi^+|_{T=0} = \frac{3k^4}{4} {\rm sech}^2\left(\frac{k\xi_+}{2}\right), \quad
\phi^-|_{T=0} = -\frac{3k^4}{4} {\rm sech}^2\left( \frac{k\xi_-}{2}\right).
\end{equation}
In what follows, we consider simplification of all expressions in the case of solitary waves with
sufficiently large $L$.

Because $\partial_T f^- = -\frac{k^2}{2} \partial_{\xi_-} f^-$,
we can integrate the linearized KdV equation (\ref{phi_minus})
in $\xi_-$, subject to the zero boundary conditions, and obtain:
$$
2 \partial_T \phi^- + \partial_{\xi_-}^3 \phi^- + \partial_{\xi_-} (f^- \phi^-) =
\partial_{\xi^-} \left( \frac{k^4}{4} f^- - k^2 \partial_{\xi_-}^2 f^- \right).
$$
Now, we recall that $f^-$ solves the stationary KdV equation:
$$
- k^2 f^- + \partial_{\xi_-}^2 f^- + \frac{1}{2} (f^-)^2 = 0,
$$
which implies that
$$
(k^2 - \partial_{\xi_-}^2 + f^-) \partial_{\xi_-} f^- = 0.
$$
Using the decomposition
$$
\phi^-(\xi_-,T) = a k^4 T \partial_{\xi_-} f^- + \psi^-(z), \quad z = \frac{k \xi_-}{2},
$$
where $a$ is a constant to be defined, we integrate the linearised equation
for $\psi^-(z)$, with zero boundary conditions, to obtain
$$
\left(4 - \partial_z^2 - 12 {\rm sech}^2(z) \right) \psi^- = 3 k^4
\left( (3 + 8a) {\rm sech}^2(z) - 6 {\rm sech}^6(z) \right).
$$
This equation is solved exactly with $\psi^- = 3 k^4 {\rm sech}^2(z)$ and
$a = -\frac{3}{8}$. Therefore, the initial-value problem for the linearized inhomogeneous KdV equation
(\ref{phi_minus}) is solved by the function
$$
\phi^-(\xi_-,T) = \tilde \phi^-(\xi_-,T)  -\frac{3}{8} k^4 T \frac{\partial f^-}{\partial \xi_-} + k^2 f^-,
$$
where $ \tilde \phi^-(\xi_-,T)$ is the solution to the Cauchy problem for the homogeneous linearized KdV equation:
\begin{equation}
2 \partial_T \tilde  \phi^- + \partial_{\xi_-}^3 \tilde \phi^- + \partial_{\xi_-} (f^- \tilde  \phi^-) = 0,
\label{phi_minus_hom}
\end{equation}
starting with the initial data
$$
\tilde \phi^-|_{T=0} = -\frac{15k^4}{4} {\rm sech}^2\left(\frac{k\xi_-}{2}\right).
$$
Note that the solutions of the linearised homogeneous KdV equations (\ref{phi_plus}) and (\ref{phi_minus_hom})
disperse to zero, so the effect of nonzero initial data decays in time.

We now compare direct numerical simulations of the Boussinesq equation (\ref{B})
with the weakly nonlinear solution (\ref{weakly-nonlinear-solution}).
We discretise the spatial domain into $N$ equally spaced points and solve
the Boussinesq equation in Matlab using a pseudo-spectral method,
based on the Fast Fourier Transform (FFT) \cite{Engelbrecht}.
The accuracy of the numerical method
is far in excess of what is required for comparisons with the weakly nonlinear
solution (\ref{weakly-nonlinear-solution}), nevertheless, extensions for further
more accurate requirements are trivially achieved by further decreasing the time
step and/or increasing the number of harmonics in the FFT. Similarly we find
the higher order terms $\phi^{\pm}$ numerically using pseudo-spectral methods
based on the FFT algorithm, analogous to the method used to solve the KdV equation in \cite{Ts} .

To consider the error of the weakly nonlinear solution, under the initial data
(\ref{BO_f+0}), we first introduce some notation. We define the numerical
solution of the Boussinesq equation (\ref{B}) as $U_{\rm num}$ and
the weakly nonlinear solutions as
\begin{eqnarray}
U_1 = f^-, \quad U_2 =  f^-   +    \epsilon \left( \phi_-^{\rm num}   +    \phi_+^{\rm num} \right),
\label{eqn:WNS_num}
\end{eqnarray}
where $\phi_{\pm}^{\rm num} $ are numerical solutions of the linearized KdV equations
(\ref{phi_plus}) and (\ref{phi_minus}). When $\phi_{\pm}$ are fixed at their initial condition
(\ref{phi-initial-data}), we denote this solution as $U_{12}$. Note that the latter solution was
previously studied in \cite{Karima1,Karima2}.

The top panels of Figure \ref{figure:WNS_BE_phi_num} depict
the evolution of the numerical solution $U_{\rm num}$
and the weakly nonlinear solutions $U_1$, $U_{12}$, and $U_2$ at times $t = 1$ (left) and $t = 1/\epsilon$ (right).
The middle panels of Figure \ref{figure:WNS_BE_phi_num} show the close-up of the area near 
the maximum of the right-propagating solitary wave at $t = 1$ (left) and of the small left-propagating wave 
at $t = 1/\epsilon$ (right). Note that the leading order approximation $U_1$ 
does not capture the generation of this left-propagating wave at all.
The bottom panels of Figure \ref{figure:WNS_BE_phi_num} also illustrate
the error terms for each of the weakly nonlinear solutions relative
to the numerical solution $U_{\rm num}$, for a particular $\epsilon$. For the time interval considered, it is clear
that there is a significant improvement in the error of the weakly nonlinear solutions 
$U_{12}$ and $U_2$ compared to the leading order approximation $U_1$.

\begin{figure}[htbp]
\begin{center}$
\begin{array}{ccc}
\quad t=1 & \quad t=1/\epsilon  \\
\includegraphics[width=2.8in
%,height=1.85in
]{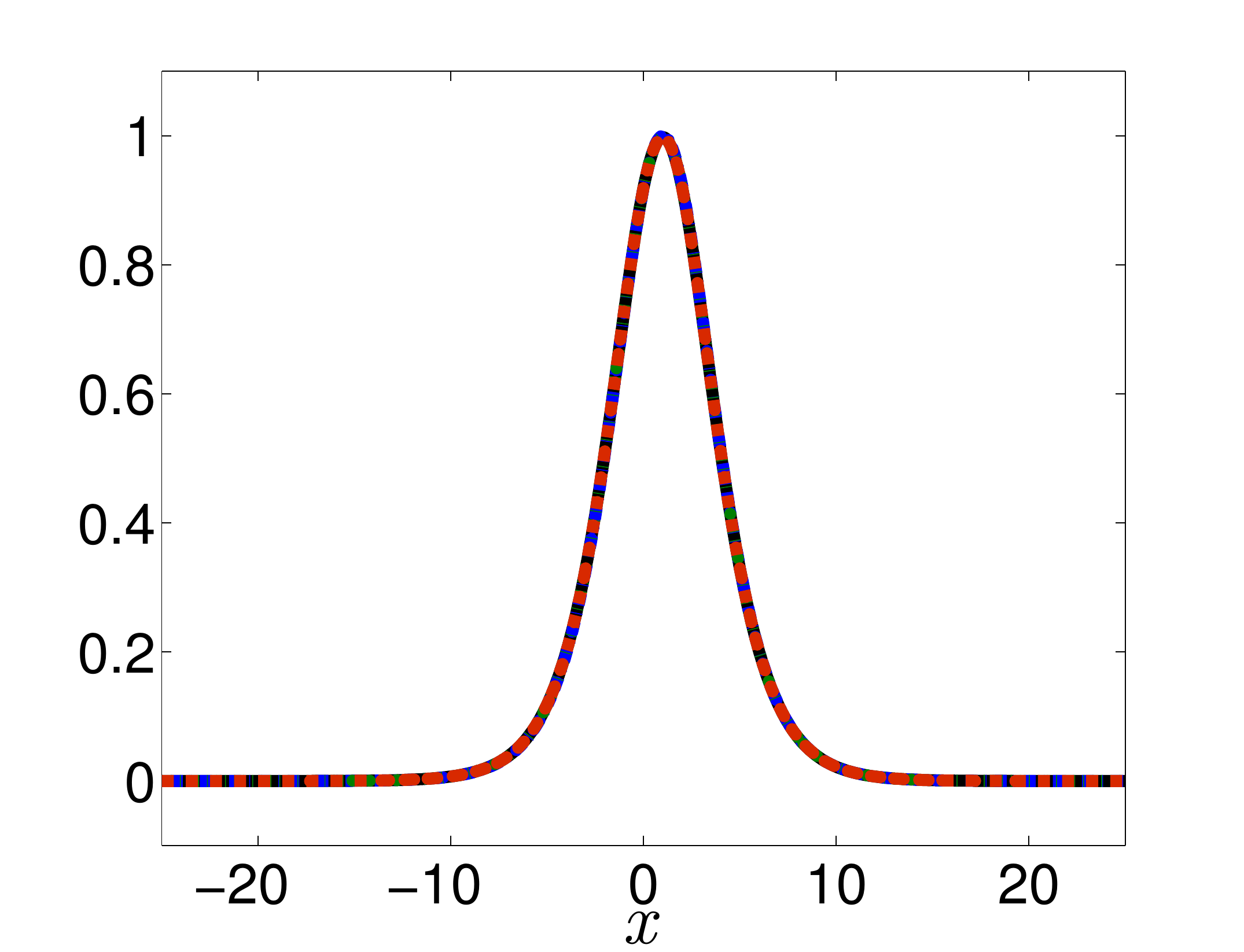} &
\includegraphics[width=2.8in
%,height=1.85in
]{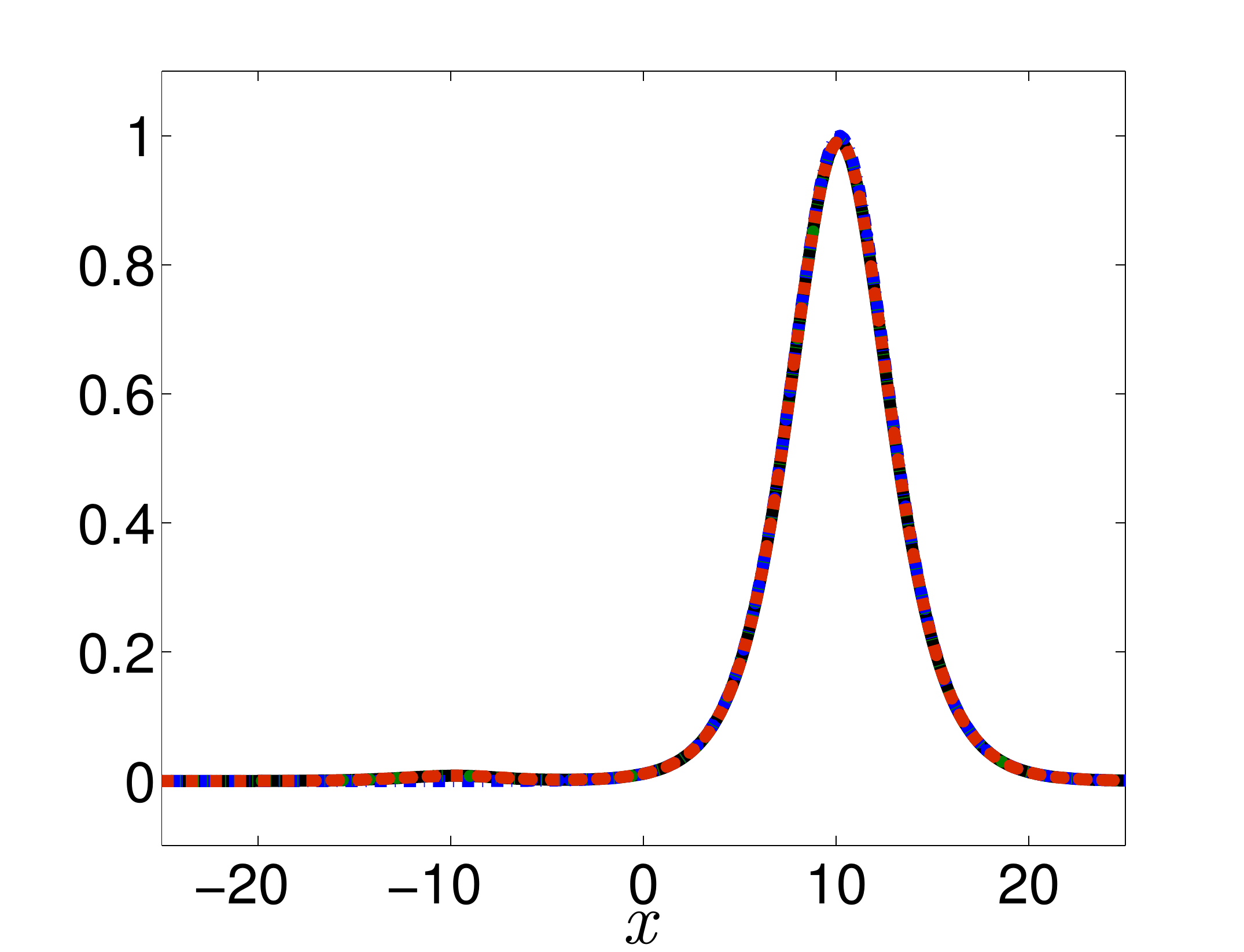}  \\
  \mbox{\footnotesize \bf (a) \scriptsize $U_{{\rm num}}$({\color{black} \bf ---}) \  $U_{1}$({\color{blue} \bf -.-}) \ $U_{12}$({\color{KMgreen} \bf $\cdot\cdot\cdot$}) \  $U_{2}$ ({\color{Brown} \bf - -}) }   &
  \mbox{\footnotesize \bf (b) \scriptsize $U_{{\rm num}}$({\color{black} \bf ---}) \  $U_{1}$({\color{blue} \bf -.-}) \ $U_{12}$({\color{KMgreen} \bf $\cdot\cdot\cdot$}) \  $U_{2}$ ({\color{Brown} \bf - -}) }     \\
\includegraphics[width=2.8in
%,height=1.85in
]{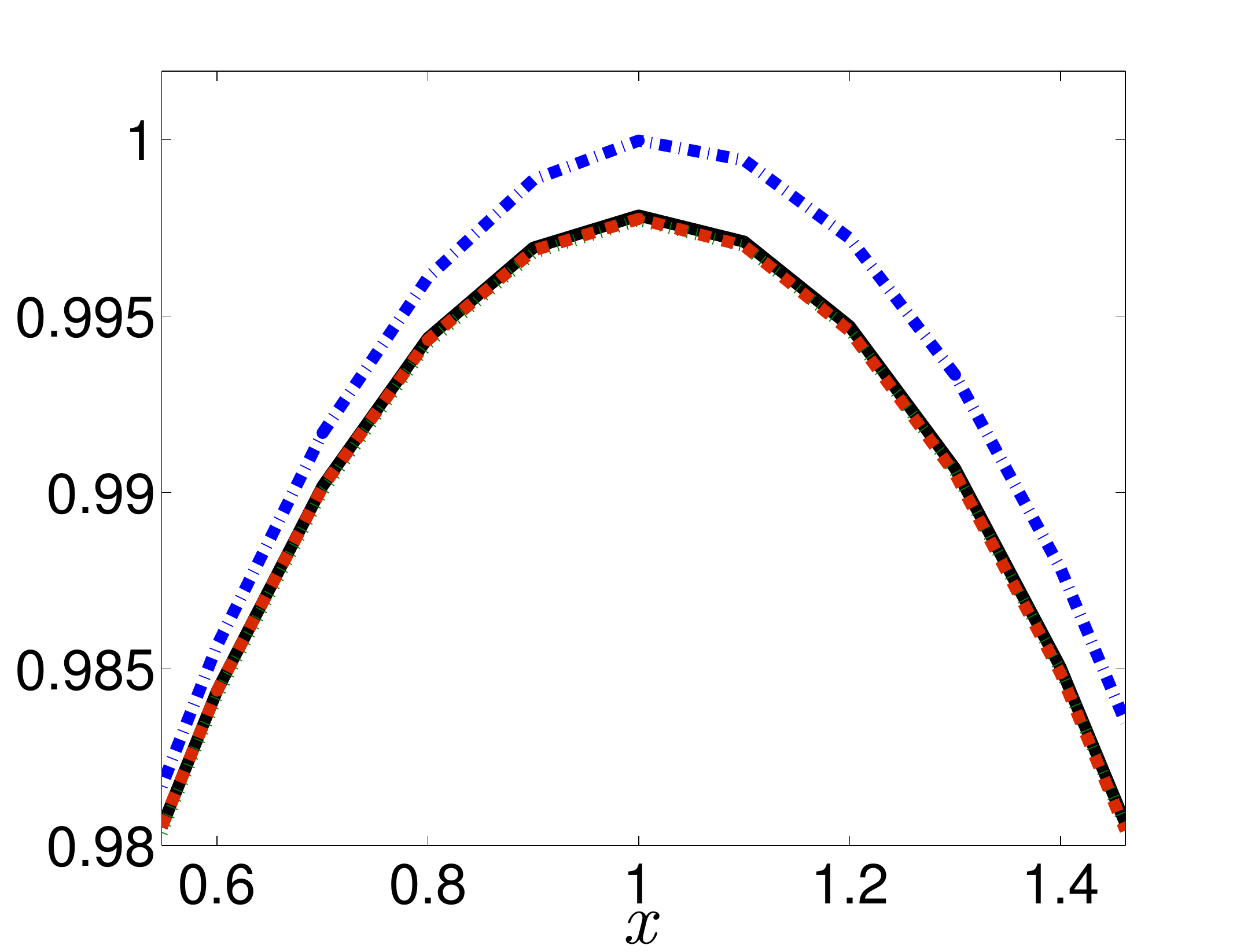} &
\includegraphics[width=2.8in
%,height=1.85in
]{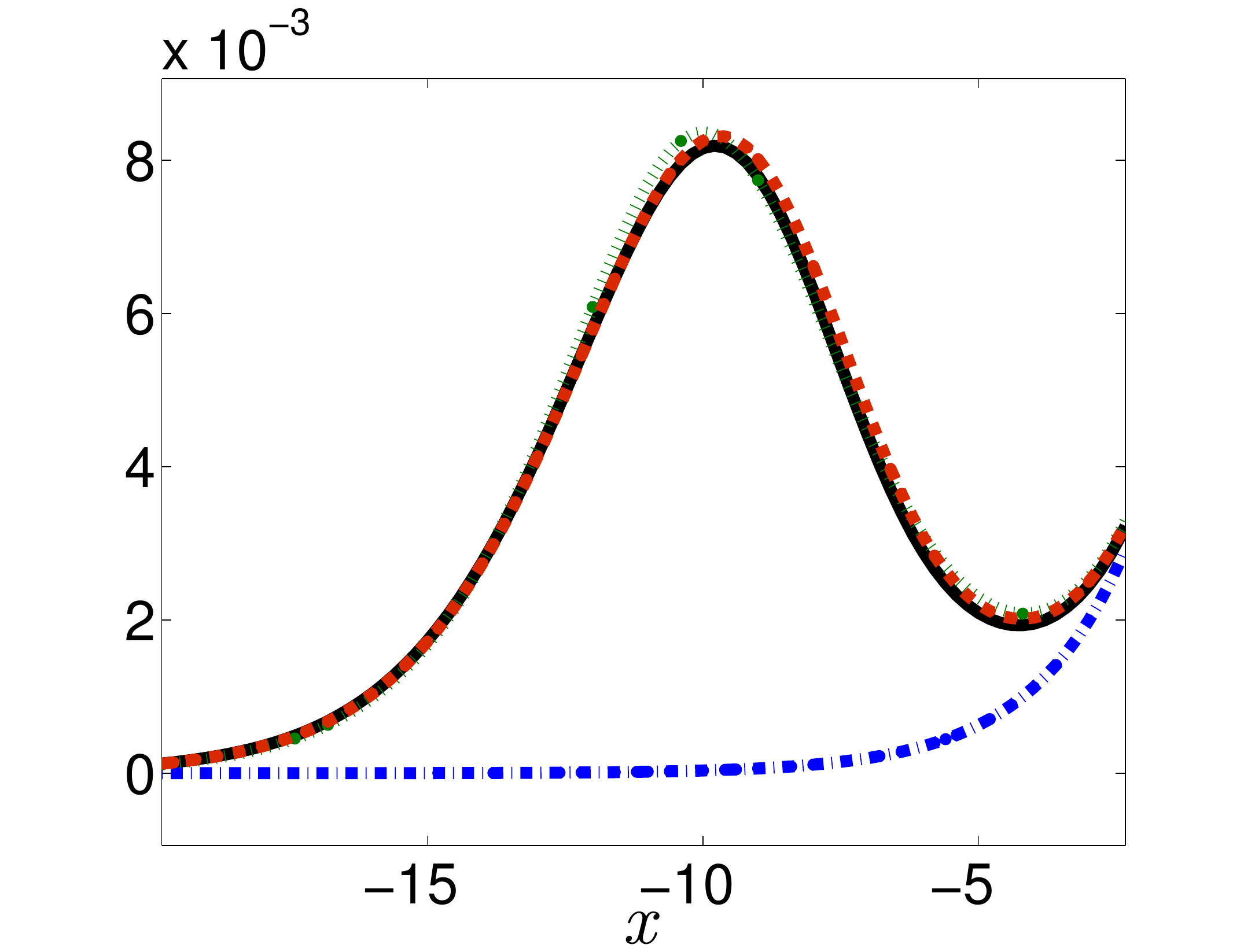}  \\
  \mbox{\footnotesize \bf (c) \scriptsize $U_{{\rm num}}$({\color{black} \bf ---}) \  $U_{1}$({\color{blue} \bf -.-}) \ $U_{12}$({\color{KMgreen} \bf $\cdot\cdot\cdot$}) \  $U_{2}$ ({\color{Brown} \bf - -}) }   &
  \mbox{\footnotesize \bf (d) \scriptsize $U_{{\rm num}}$({\color{black} \bf ---}) \  $U_{1}$({\color{blue} \bf -.-}) \ $U_{12}$({\color{KMgreen} \bf $\cdot\cdot\cdot$}) \  $U_{2}$ ({\color{Brown} \bf - -}) }     \\
\includegraphics[width=2.8in
%,height=2.15in
]{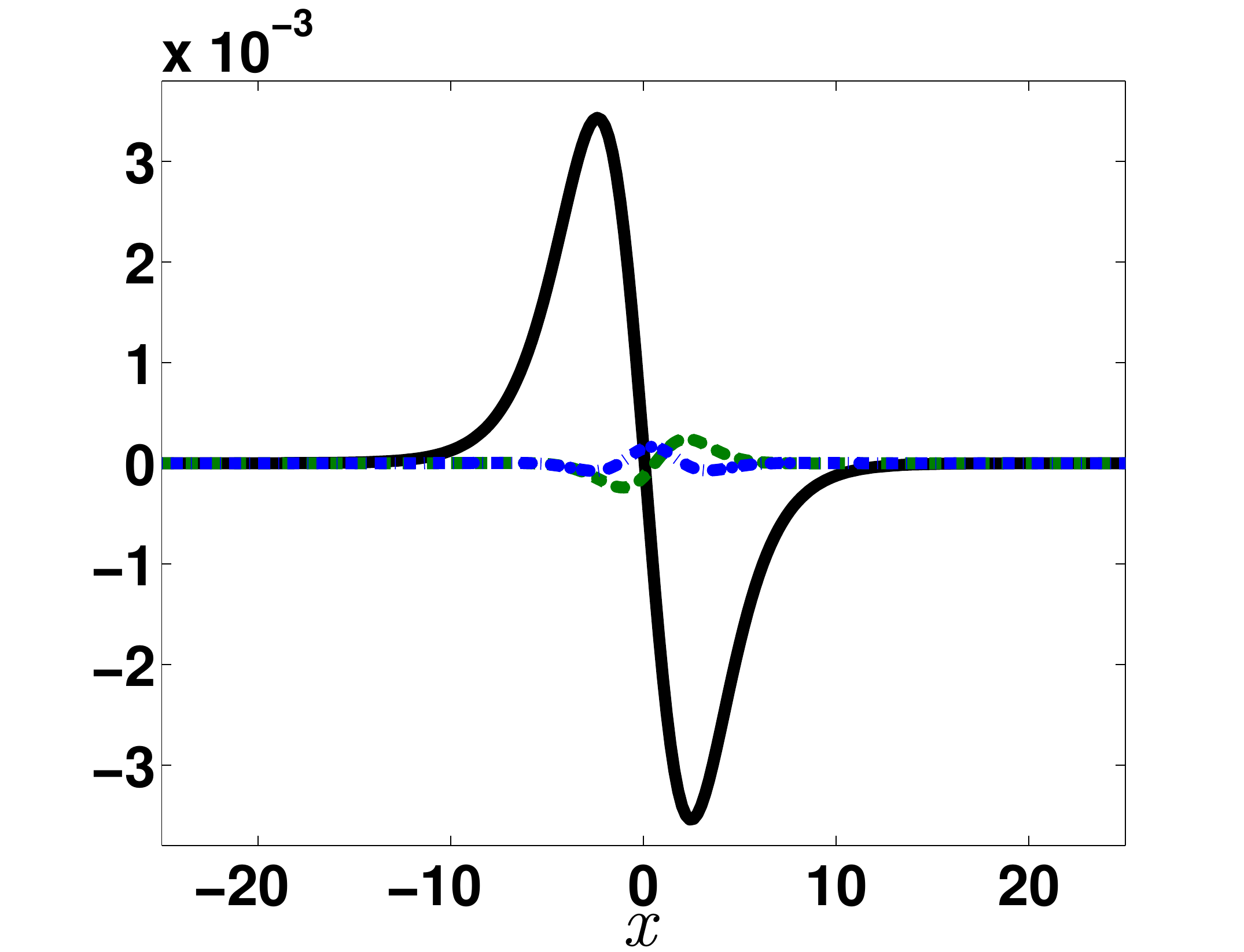}  &
\includegraphics[width=2.8in
%,height=2.08in
]{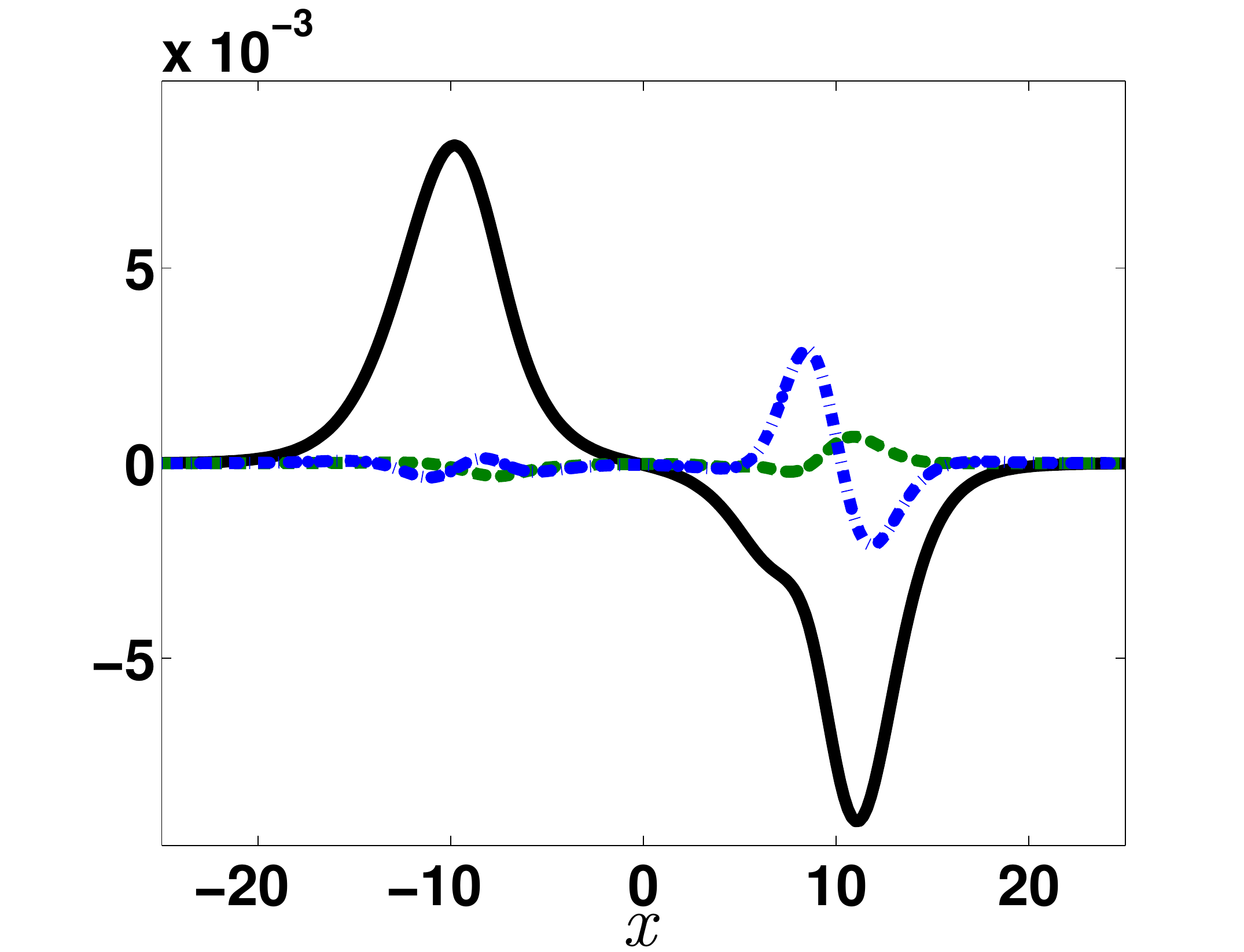} \\
\mbox{\footnotesize \bf (e) \scriptsize $(U_{{\rm num}} - U_{1})$({\color{black} \bf ---}) \  $(U_{{\rm num}} - U_{12})$({\color{blue} \bf -.-})   }   &
\mbox{\footnotesize \bf (f) \scriptsize $(U_{{\rm num}} - U_{1})$({\color{black} \bf ---}) \  $(U_{{\rm num}} - U_{12})$({\color{blue} \bf -.-})  }   \\
\mbox{\footnotesize \bf  \scriptsize   $(U_{{\rm num}} - U_2)$({\color{KMgreen} \bf - -}) }   &
\mbox{\footnotesize \bf  \scriptsize   $(U_{{\rm num}} - U_2)$({\color{KMgreen} \bf - -}) }
\end{array}$
\end{center}
\vspace{-3mm}
\caption{\small Comparison of the weakly nonlinear solutions $U_1$, $U_{12}$, and $U_2$
with the numerical solution $U_{\rm num}$ for $k=1/\sqrt{3}$, $\epsilon=0.1$, $\gamma = 0$
at (a) $t = 1 \ \& \ $(b) $t = 1/\epsilon$, with the close-up of some areas (c) \& (d) and the error plots (e) \& (f) at the respective times.
Numerical parameters: $\Delta t = 0.01$, $\Delta T = 0.00125$ and $L=2000$, $N=2\times10^4$.}
 \label{figure:WNS_BE_phi_num}
\end{figure}

The important question is now: how does the error of the solutions $U_1$, $U_{12}$ and $U_2$
scale with $\epsilon$? To analyse the error of the weakly nonlinear solutions in more detail,
we consider the maximum absolute error over $x$ defined as
\begin{eqnarray}
e^{i}_t = \max_{-L \leq x \leq L}  | U_{\rm num}(x,t) - U_{i} (x,t)|, \quad i = 1, 12, 2.
\label{eqn:WNS_pert_num_error}
\end{eqnarray}
We use a least squares power fit to determine how the maximum absolute error of
the weakly nonlinear solution at each order varies with the small parameter $\epsilon$.
First we write the maximum errors defined in (\ref{eqn:WNS_pert_num_error}) in the form
\begin{eqnarray}
e^{i}_{t} &=& C_i \epsilon^{\alpha_i}, \quad {\rm for}  \quad i=1, 12, 2,\
\label{eqn:WNS_pert_num_error_rewritten}
\end{eqnarray}
corresponding to each order of $\epsilon$ in the  weakly nonlinear solution (\ref{eqn:WNS_num}).
Taking logs of the errors in this form and considering a range of $\epsilon$, one can 
find the coefficients $C_i$ and $\alpha_i$, with the latter revealing how the maximum absolute 
errors scale with $\epsilon$. We find the coefficients using Matlab's ``{\it polyfit}" command.

\begin{figure}[h]
\begin{center}$
\begin{array}{cc}
\quad t=1 & \quad t=1/\epsilon  \\
\includegraphics[width=3.4in]{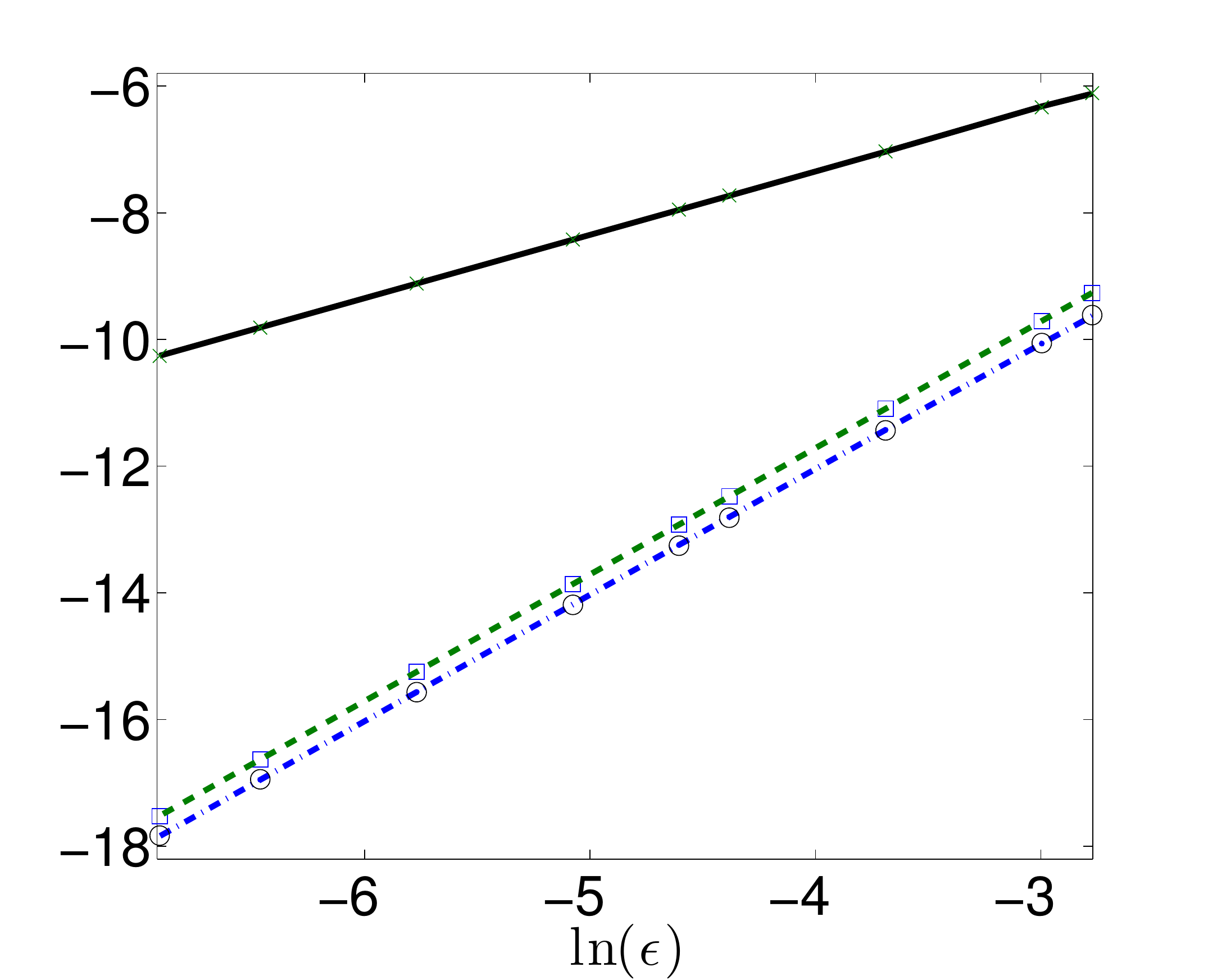}
&
\includegraphics[width=3.3in,height=2.75in]{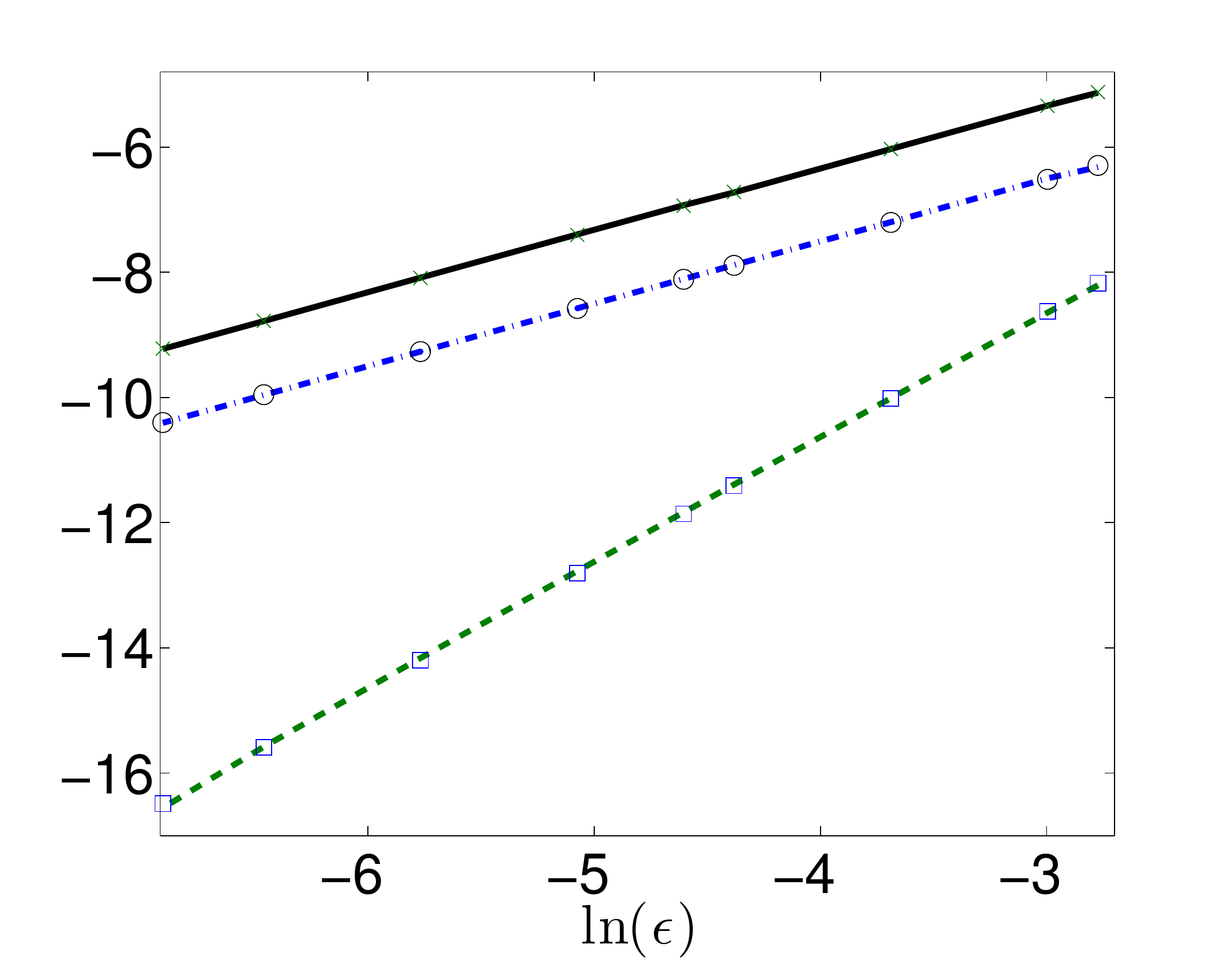}
\\
\multicolumn{2}{c}{  \quad  \mbox{\bf \scriptsize
${\rm ln}(e^{1}_{t})$({\color{black} \bf ---}) \ \ \ \  $\alpha_{1} {\rm ln}(\epsilon) + {\rm ln}(C_{1})$({\color{KMgreen}  $\times \ \times$}) \ \ \ \
${\rm ln}(e^{12}_{t})$({\color{blue} \bf -.-}) \ \ \ \  $\alpha_{12} {\rm ln}(\epsilon) + {\rm ln}(C_{12})$({\color{black} $\bigcirc \ \bigcirc$})
} }\\
\multicolumn{2}{c}{  \quad  \mbox{\bf \scriptsize
${\rm ln}(e^{2}_{t})$({\color{KMgreen} \bf - -}) \ \ \ \    $\alpha_{2} {\rm ln}(\epsilon) + {\rm ln}(C_{2})$({\color{blue} $\square \ \square$})
}} \\
\vspace{-7mm}
\end{array}$
\end{center}
\caption{\small Double log plot of absolute errors versus $\epsilon$
at $t =1$ and $t =1/\epsilon$. 
Numerical parameters are the same as in Figure \ref{figure:WNS_BE_phi_num}. 
Coefficients $\alpha_{1,12,2}$ and $C_{1,12,2}$ are
given in Table \ref{tab:error_scale}.}
\label{figure:error_eps_T=1}
\end{figure}
%
% This is done for 9 points [0.0625,0.001]

%% 3D plot uncomment to include
%\begin{figure}[h]
%\begin{center}$
%\begin{array}{c}
%
%
%\includegraphics[width=3.5in]{BE_27-07-13_3Dplot}\\
%
%
%\vspace{-7mm}
%\end{array}$
%\end{center}
%\caption{\small Evolution of the weakly nonlinear solution $U$ for $k=1/\sqrt{3}$ and $\epsilon=0.1$.}
%\label{figure:BE_3D}
%\end{figure}

%
%
%
%\begin{figure}[htbp]
%\begin{center}
%\includegraphics[width=3.5in]{ex1_3D.pdf}
%\end{center}
%\caption{\small 3D plot of the solution}
%\label{fig:BE_3D}
%\end{figure}
%
%

In Figure \ref{figure:error_eps_T=1} we display double log plots of the maximum
absolute errors, which we find explicitly from simulations
against $\epsilon$, and the log of the errors as
defined in (\ref{eqn:WNS_pert_num_error_rewritten}), both at times $t = 1$ (left) and $t=1/\epsilon$ (right).
As depicted in Figure \ref{figure:error_eps_T=1}, we find the weakly nonlinear
solution $U_2$ dramatically improves the scaling of the maximum absolute error at the time 
$t = 1/\epsilon$, in comparison with the solution $U_{12}$. Furthermore, the maximum absolute error of $U_2$
scales almost precisely as $\mathcal{O}(\epsilon^2)$, in agreement with the rigorous error estimates
of Theorem \ref{theorem-main}.

We note from Figure \ref{figure:error_eps_T=1} the strong similarity
in the error scalings of the leading order solution $U_1$
with the higher order solution $U_{12}$ at $ t = 1 / \varepsilon$.
This was not observed in \cite{Karima1,Karima2} since the detailed analysis of the error
term was not undertaken. However, it must be noted, upon direct comparison of the
two solutions, that the maximum absolute error of the  solution
$U_{12}$  is still significantly smaller than the error of the solution $U_1$.
This is because the constant $C_{12}$ in the  error term (\ref{eqn:WNS_pert_num_error_rewritten})
is substantially smaller than the corresponding leading order constant $C_1$ 
(see Table \ref{tab:error_scale}).

\begin{table*} [htbp]
\centering
\begin{tabular}{c |  c  c  c  c }  \hline \vspace{-4.5mm} \\
&\multicolumn{2}{c}{{\ \ \rm Boussinesq equation} \ \ \ \ \ \ \ \ \ \ \ \ } &  \multicolumn{2}{c}{ \rm Boussinesq--Ostrovsky equation}  \\ \hline
 & {\rm $t=1$}  & {\rm $t=1/\epsilon$}   & {\rm  $t=1$}   & {\rm $t=1/\epsilon$}  \\ \hline
$\alpha_1$& \  \  1.0032  \ \   &  \  \  0.9906   \  \  &  1.0026  &    0.9690    \\
$\alpha_{12}$&  1.9870 & 0.9928                  &  \  \  ---   \ \ & \ \ --- \ \ \\
$\alpha_{2}$& 1.9970 & 2.0103                     &  1.9302  & 1.9515 \\
$C_1$       &  0.03576 & 0.04979                   &  0.1668 &   0.4851  \\
$C_{12}$ & 0.01648 & 0.02896                    &  ---  & --- \\
 $C_2$     & 0.02407 & 0.07417                  & 0.5474  &  9.2062  \\
%0.0125& 3.9559 $\rm{x} \ 10^{-6}$ & 9.0239 ${\rm x \ 10^{-7}}$ & 2.8532 ${\rm x \ 10^{-5}}$  & 1.1751 ${\rm x \ 10^{-5}}$  \\
%0.00625& 4.0198 $\rm{x} \ 10^{-6}$ & 3.6362 ${\rm x \ 10^{-6}}$ & 2.8988 ${\rm x \ 10^{-5}}$  & 4.6968 ${\rm x \ 10^{-5}}$  \\
\hline
\end{tabular}
\caption{\small Maximum absolute error scaling parameters corresponding to results illustrated in Figures \ref{figure:error_eps_T=1} \& \ref{figure:error_eps_BOE_T=1}.}
\label{tab:error_scale}
\end{table*}

\subsection{Boussinesq--Ostrovsky equation}
%¤
%
Let us now consider an initial-value problem for the regularized Boussinesq--Ostrovsky equation (\ref{BO})
starting with initial data
\begin{eqnarray}
\left\{ \begin{array}{l}
U|_{t=0} = 3k^2 \ {\rm sech}^2 (\frac{kx}{2}) -  \hat{\alpha} [ {\rm sech}^2\left  (\frac{k(x+x_0)}{2}\right ) +  {\rm sech}^2 \left  (\frac{k(x-x_0)}{2} \right ) ],\\
U_t|_{t=0} =  3k^3 \ {\rm sech}^2 (\frac{kx}{2})  \ {\rm tanh} (\frac{kx}{2})
  -  k \hat{\alpha}
\Big[ {\rm sech}^2 \left ( \frac{k(x+x_0)}{2} \right ) \ {\rm tanh} \left (\frac{k(x+x_0)}{2} \right )  \\
\hspace{1.2cm} +  {\rm sech}^2 \left (\frac{k(x-x_0)}{2} \right ) \ {\rm tanh} \left (\frac{k(x-x_0)}{2} \right ) \Big], \end{array} \right.
\label{BO_2}
\end{eqnarray}
where $x_0$ is an arbitrary shift along $x$. The initial data is defined
on the periodic domain $-L \leq x \leq L$. If we choose $\hat{\alpha}$ to be 
$$
\hat{\alpha} = \frac{3k^2 \ {\rm tanh}(kL/2) }{   {\rm tanh}(k(L+x_0)/2) +  {\rm tanh}(k(L-x_0)/2)   },
$$
then the mean value of $U|_{t = 0}$ is zero, $c_0 = 0$. 

The weakly nonlinear solution is still given by the expansion (\ref{weakly-nonlinear-solution})
but the leading order term $f^-$ is now a solution to the Ostrovsky equation (\ref{KdV-Ost})
starting with the initial data
\begin{eqnarray}
f^-|_{T=0} = 3k^2 \ {\rm sech}^2 \left (\frac{k \xi^- }{2}\right ) -   \hat \alpha    \Big[ {\rm sech}^2 \left (\frac{k(\xi^-+x_0)}{2} \right)
+ {\rm sech}^2 \left (\frac{k(\xi^- -x_0)}{2} \right ) \Big].
\label{BOE_f_minus}
\end{eqnarray}
The higher order terms $\phi^{\pm}$ satisfy the linearised Ostrovsky equations
\begin{eqnarray}
 \partial_{\xi_+} \phi^+  \left (-2 \partial_T \phi^+ + \partial_{\xi_+}^3 \phi^+ \right ) = \gamma \phi^+,
\label{BOE_phi_plus}
\end{eqnarray}
and
\begin{eqnarray}
\partial_{\xi^-} \left( 2 \partial_T \phi^- +
\partial_{\xi_-}^3 \phi^- + \partial_{\xi_-} (f^- \phi^-) \right)
= \gamma \phi^-   + \partial_T^2 f^- + 2 \partial_T \partial_{\xi_-}^3 f^-,
\label{BOE_phi_minus}
\end{eqnarray}
starting with the initial data
\begin{eqnarray}
\left\{ \begin{array}{l}
\phi^+|_{T=0} = - \Phi(\xi_+), \\
\phi^-|_{T=0} =  \Phi(\xi_-), \end{array} \right . ,
\end{eqnarray}
where $\Phi$ can be expressed in terms of the leading order solution $f^-$ as
\begin{equation}
\Phi(x) = \frac{1}{2} \left( \int^{x}_{-L} f^-_T(s)ds -
\frac{1}{2L} \int^{L}_{-L} \left(\int^{y}_{-L}  f^-_T(s)ds  \right)dy    \right) \biggr|_{T=0}.
\label{Phi}
\end{equation}
The function $f^-$ is a solution to the Ostrovsky equation (\ref{KdV-Ost}), 
and therefore the derivative $f^-_{T}$ in (\ref{Phi}) can be readily expressed as
$$
f^-_{T} = \frac 12 \left( - f^- \partial_{\xi^-} f^- - \partial_{\xi^-}^3 f^- + \gamma \partial_{\xi^-}^{-1} f^- \right).
$$
Then, the formula (\ref{Phi}) shows, in particular, that the magnitude of the higher order 
corrections will be smaller if the initial data for the function $f^-$ is localised, 
and for smaller values of the constant $\gamma$. To simplify our numerical 
 simulation, we choose the initial data accordingly.

The Cauchy problems for equations (\ref{KdV-Ost}), (\ref{BOE_phi_plus}), and (\ref{BOE_phi_minus}) are solved
numerically and simultaneously at each time step.
The Ostrovsky equation (\ref{KdV-Ost}) can be solved numerically using both 
pseudo-spectral and finite difference methods (e.g. \cite{GH,Miyatake,OS,Yaguchi}). 
We extend the spectral method used to solve the KdV equation in \cite{tref} 
in order to solve the linearised Ostrovsky equations (\ref{BOE_phi_plus}) and (\ref{BOE_phi_minus}). 
The method used to solve the Boussinesq-Ostrovsky equation is an extension to the 
numerical method in \cite{Engelbrecht}.

We now compare direct numerical simulations
of the Cauchy problem for the Boussinesq--Ostrovsky equation (\ref{BO}) with the weakly nonlinear solutions
$U_1$ and $U_2$ defined by (\ref{eqn:WNS_num}), as well as analyze the approximation error from
(\ref{eqn:WNS_pert_num_error}) and (\ref{eqn:WNS_pert_num_error_rewritten}).

The top panels of Figure \ref{figure:WNS_num_BOE} depict
the evolution of the numerical solution $U_{\rm num}$
and the weakly nonlinear solutions $U_1$ and $U_2$ at times $t = 1$ (left) and $t = 1/\epsilon$ (right).
The middle panels of Figure \ref{figure:WNS_num_BOE} show the close-up of the area 
near the maximum of the right-propagating wave at $t = 1$ (left) and of the small left-propagating wave 
at $t = 1/\epsilon$ (right). We again note, similarly to our first example, that the leading order approximation 
$U_1$ does not capture the generation of this left-propagating wave at all. We also note the important 
qualitative change in the dynamics of the solution, compared to the Boussinesq equation: 
the initial data generates  the right-propagating nonlinear wave packet instead of a solitary wave, 
which agrees with the well known behaviour of solutions of the Ostrovsky equation \cite{GH}.
The bottom panels of Figure \ref{figure:WNS_num_BOE} illustrate
the evolution of errors for each of the weakly nonlinear solutions $U_1$ and $U_2$ relative
to the numerical solution $U_{\rm num}$, for a particular $\epsilon$.
One can notice a distinct improvement in the accuracy of the  weakly nonlinear solution $U_2$
compared with $U_1$.

%
%\begin{figure}[htbp]
%\begin{center}$
%\begin{array}{ccc}
%\quad t=1 & \quad t=1/\epsilon  \\
%
%
%
%\includegraphics[width=2.2in
%,height=2.15in
%]{BOE_gam=1_T=eps_31_07_13.pdf}  &
%\includegraphics[width=2.25in
%,height=2.08in
%]{BOE_gam=1_T=1_31_07_13.pdf} \\
%
%
%\includegraphics[width=2.2in
%,height=2.15in
%]{BOE_gam=1_T=eps_error_31_07_13.pdf}  &
%\includegraphics[width=2.3in
%,height=2.08in
%]{BOE_gam=1_T=1_error_31_07_13.pdf} \\
%\mbox{\footnotesize \bf  \scriptsize $(U_{{\rm num}} - U_{1})$({\color{blue} \bf ---}) \  $(U_{{\rm num}} - U_{2})$({\color{KMgreen} \bf ---})   }   &
%\mbox{\footnotesize \bf  \scriptsize $(U_{{\rm num}} - U_{1})$({\color{blue} \bf ---}) \  $(U_{{\rm num}} - U_{2})$({\color{KMgreen} \bf ---})  }   \\
%
%
%
%
%\end{array}$
%\end{center}
%\vspace{-1mm}
%\caption{\small Comparison of the weakly nonlinear solutions $U_1$ and $U_2$
%with the numerical solution $U_{\rm num}$, for $k=1/\sqrt{3}$, $\epsilon=0.001$,
%$\gamma=1$, $x_0=10$, at (a) $t = 1$ \& (b) $t = 1/\epsilon$, and the error plots
%at the respective times. Numerical parameters: $\Delta t = 0.01$,
%$\Delta T = 0.00125$ and $L=40$, $N=2\times10^3$.}
%% \label{figure:WNS_BOE_comparison}
%\end{figure}
%
%
%
%
%
\begin{figure}[htbp]
\begin{center}$
\begin{array}{ccc}
\quad t=1 & \quad t=1/\epsilon  \\
\includegraphics[width=2.8in
%,height=1.85in
]{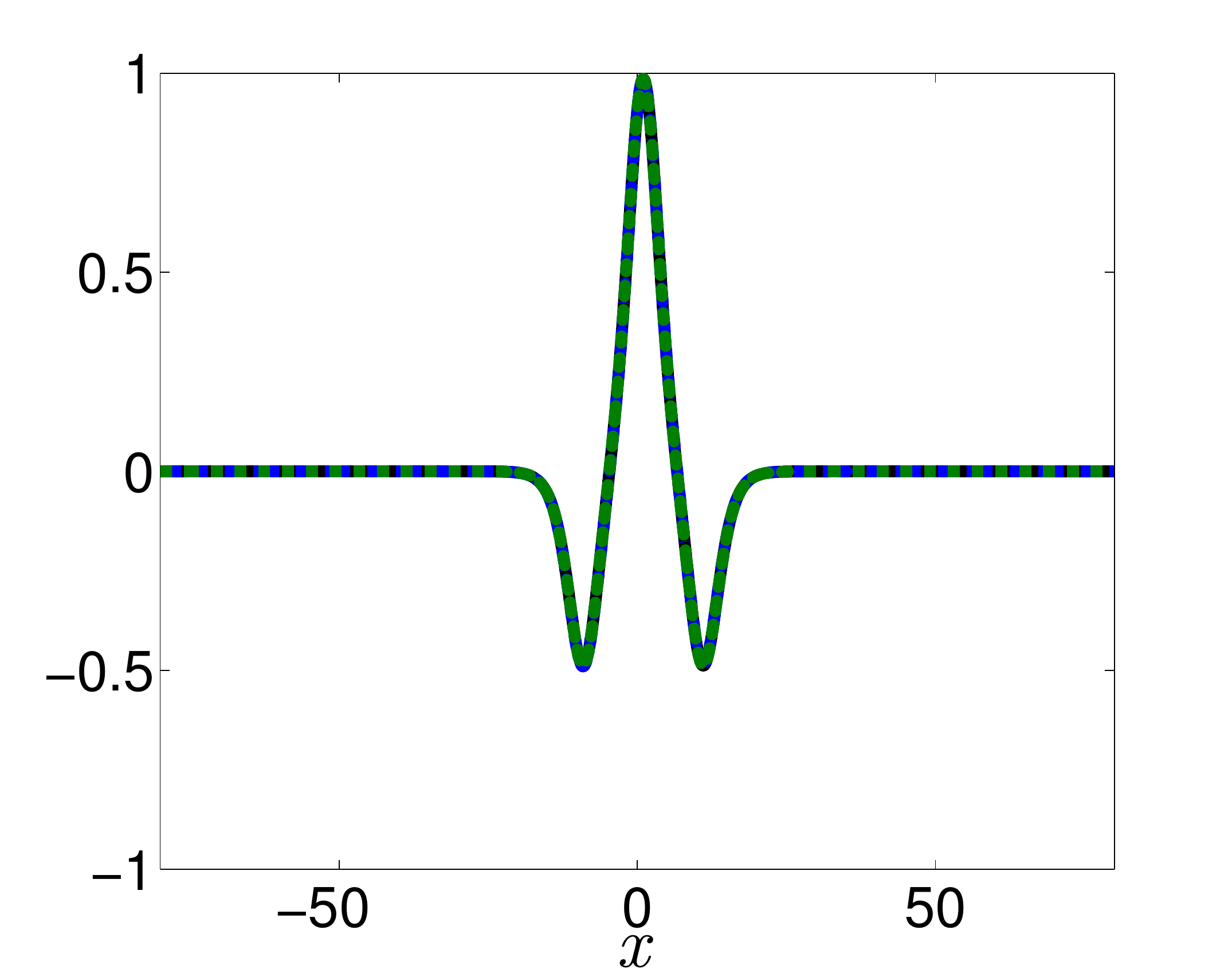} &
\includegraphics[width=2.8in
%,height=1.85in
]{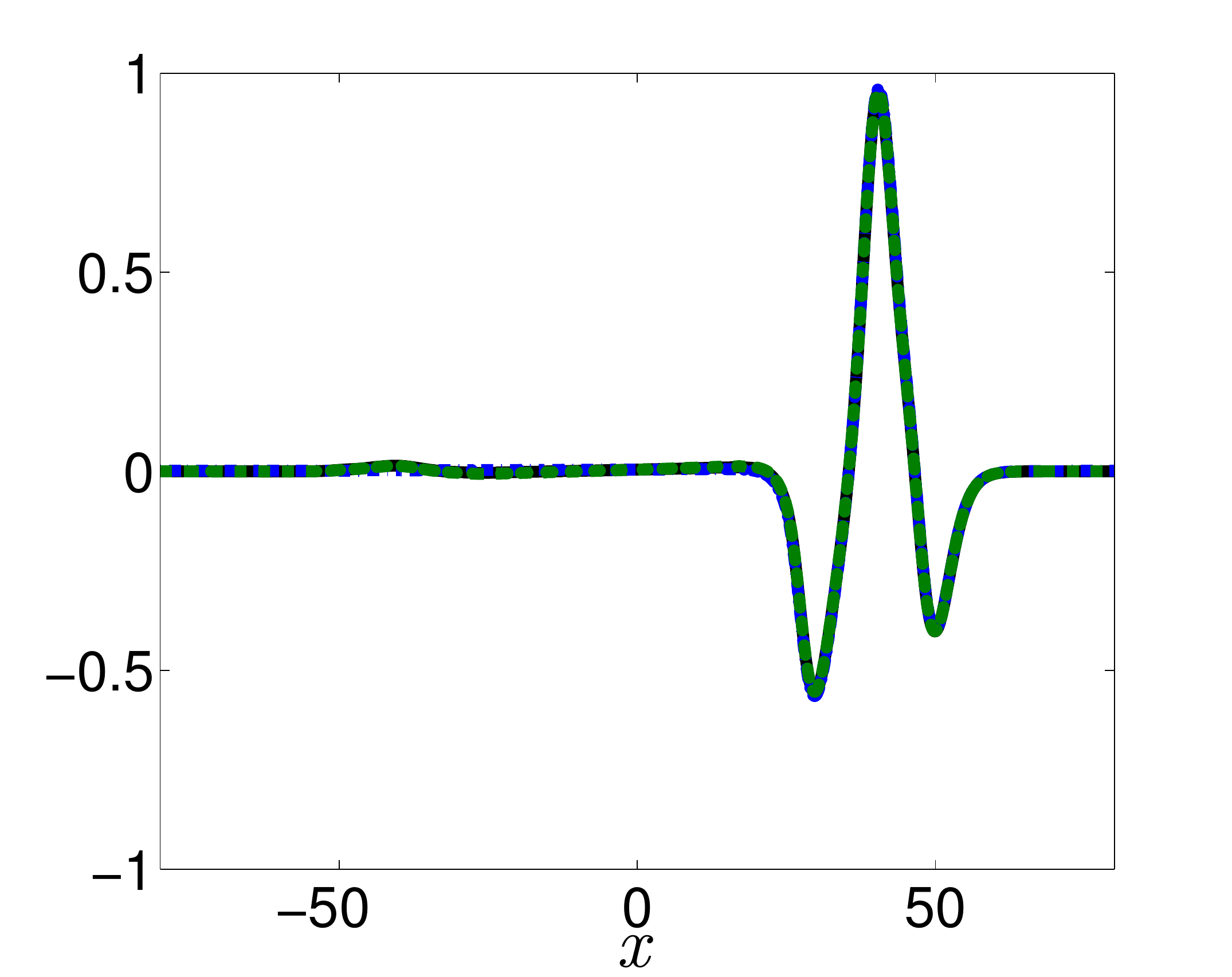}  \\
  \mbox{\footnotesize \bf (a) \scriptsize $U_{{\rm num}}$({\color{black} \bf ---}) \  $U_{1}$({\color{blue} \bf -.-}) \ $U_{2}$({\color{KMgreen} \bf - -}) }  &
  \mbox{\footnotesize \bf (b) \scriptsize $U_{{\rm num}}$({\color{black} \bf ---}) \  $U_{1}$({\color{blue} \bf -.-}) \ $U_{2}$({\color{KMgreen} \bf - -}) }    \\
\includegraphics[width=2.8in
%,height=1.85in
]{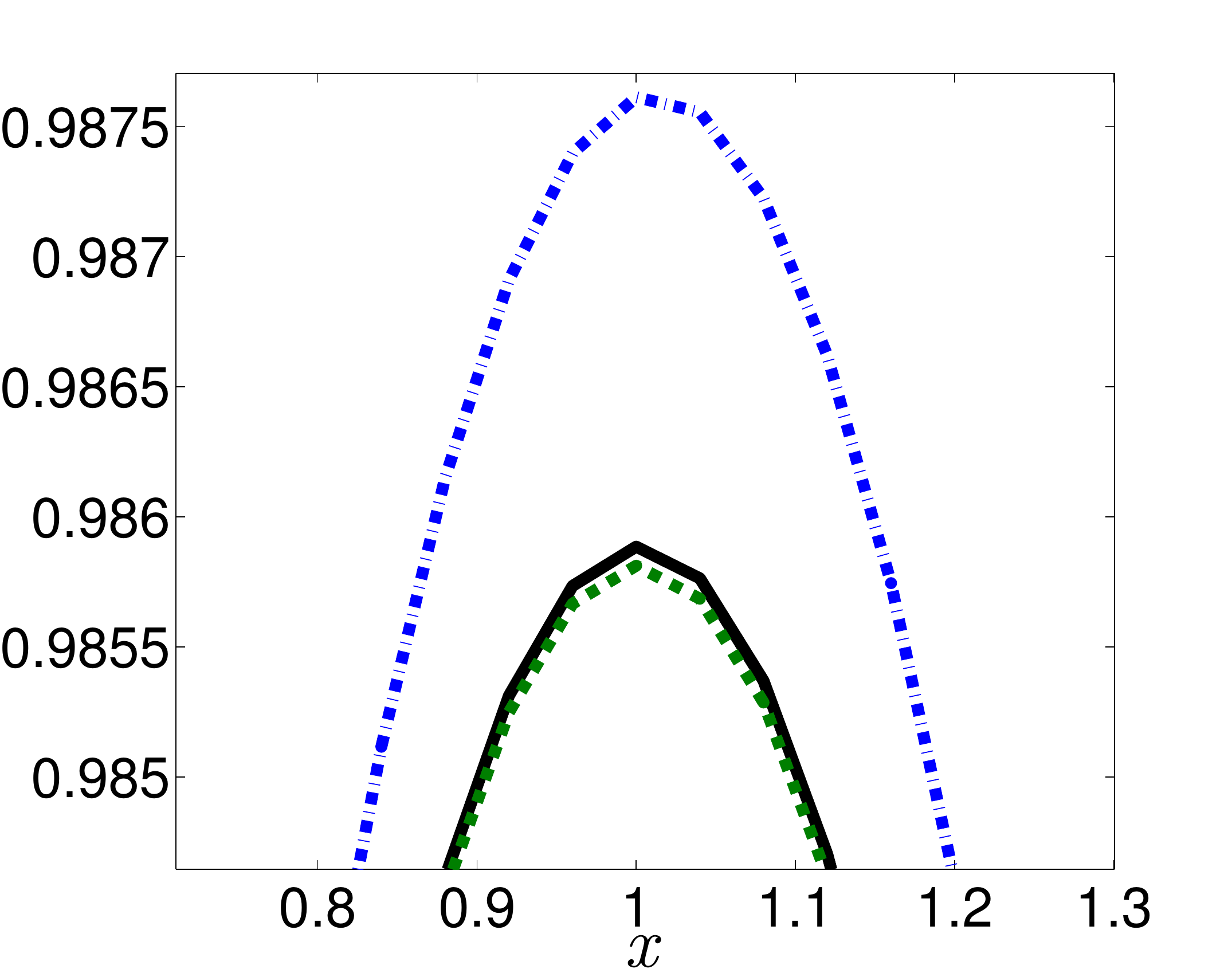} &
\includegraphics[width=2.8in
%,height=1.85in
]{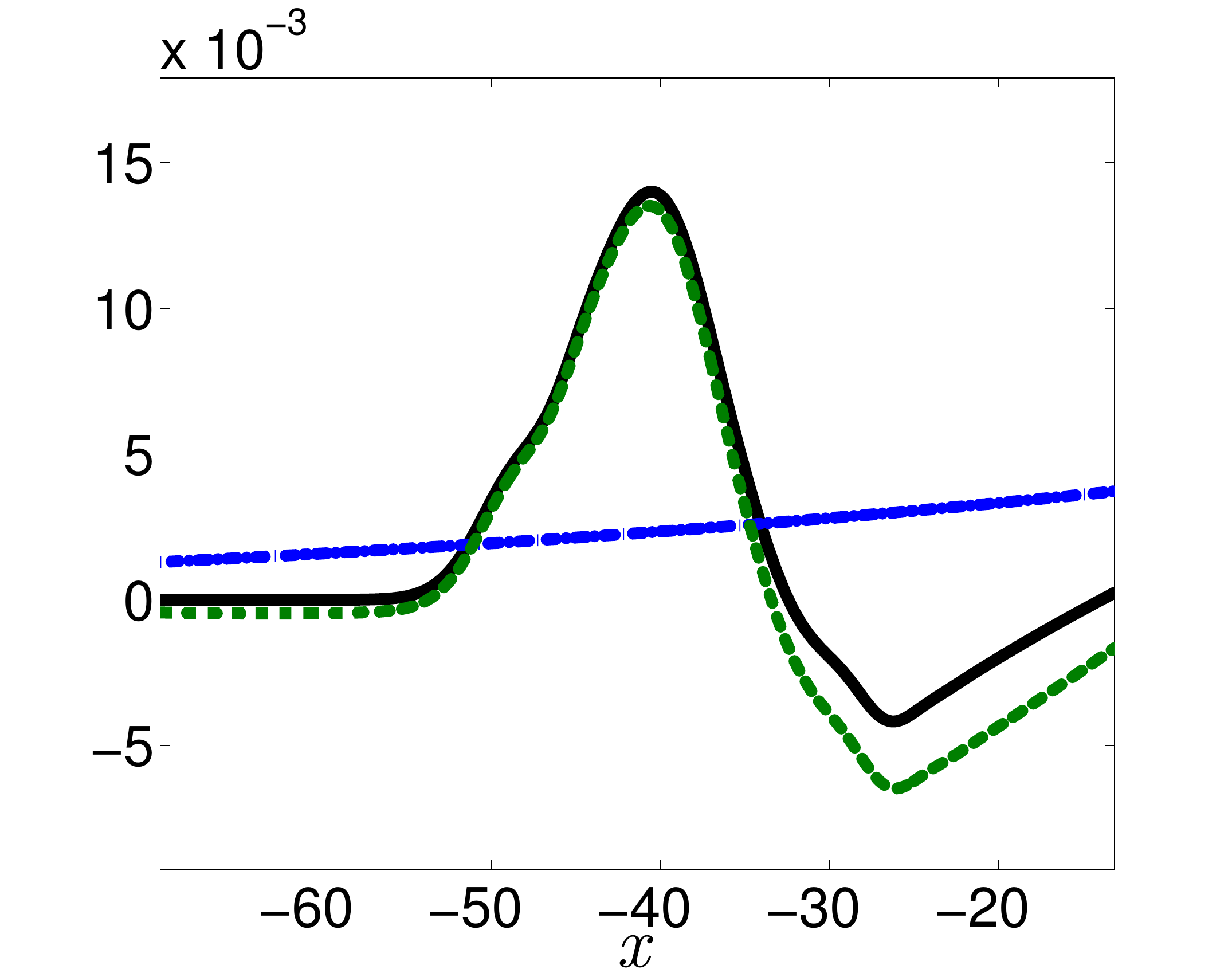}  \\
  \mbox{\footnotesize \bf (c) \scriptsize $U_{{\rm num}}$({\color{black} \bf ---}) \  $U_{1}$({\color{blue} \bf -.-}) \ $U_{2}$({\color{KMgreen} \bf - -}) }&
  \mbox{\footnotesize \bf (d) \scriptsize $U_{{\rm num}}$({\color{black} \bf ---}) \  $U_{1}$({\color{blue} \bf -.-}) \ $U_{2}$({\color{KMgreen} \bf - -}) }   \\
\includegraphics[width=2.8in
%,height=2.15in
]{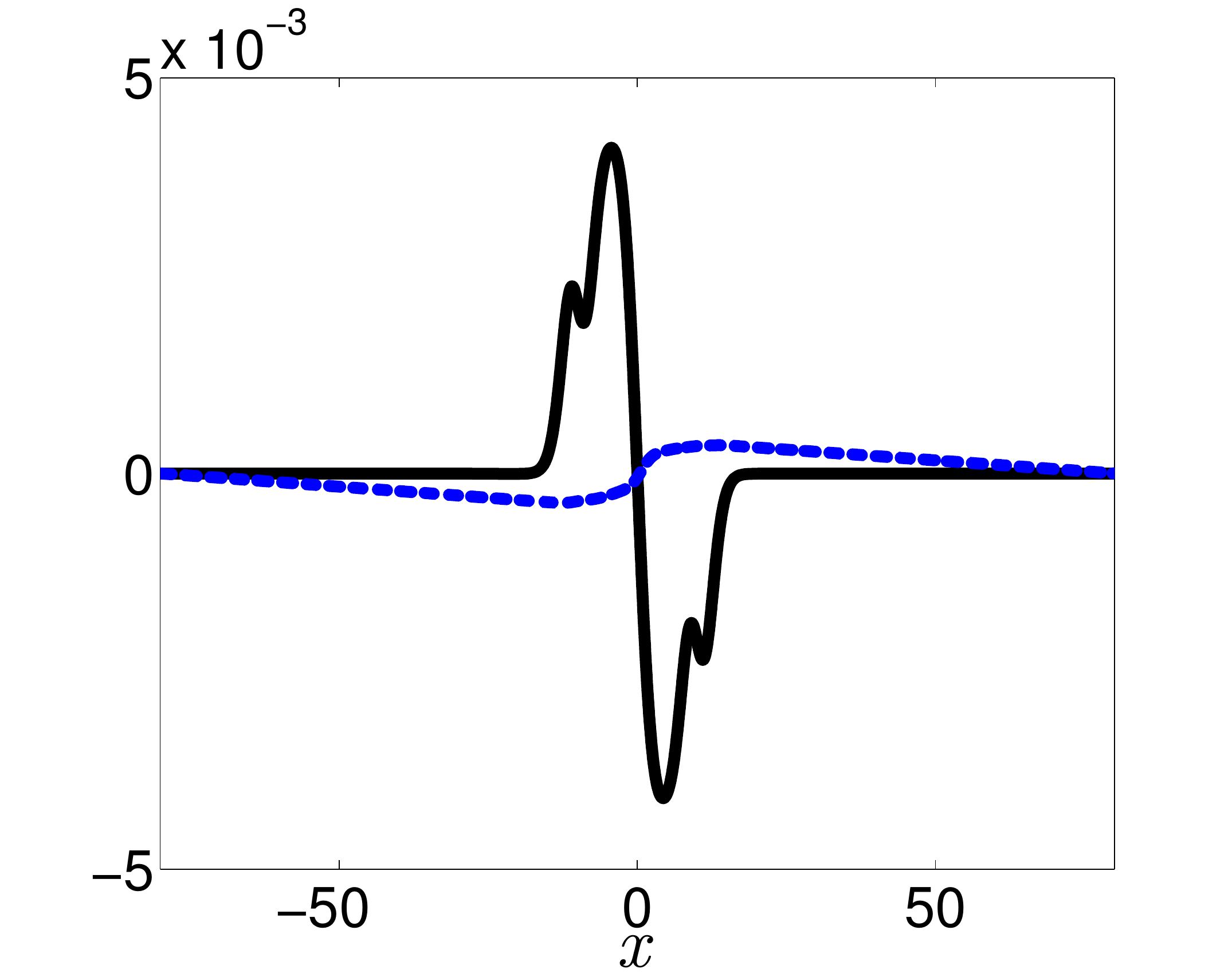}  &
\includegraphics[width=2.8in
%,height=2.08in
]{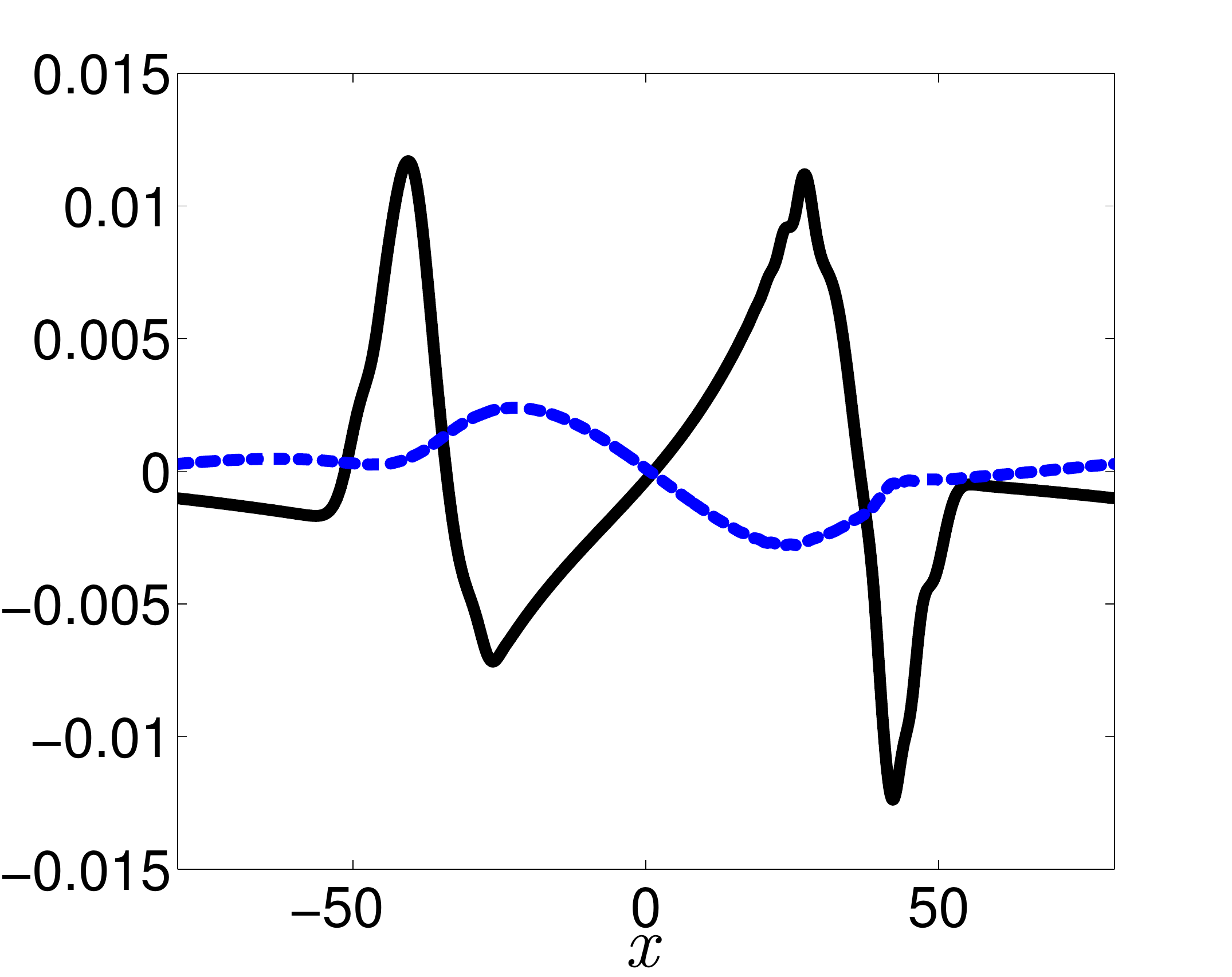} \\
\mbox{\footnotesize \bf (e) \scriptsize $(U_{{\rm num}} - U_{1})$({\color{black} \bf ---}) \  $(U_{{\rm num}} - U_{2})$({\color{blue} \bf - -})   }   &
\mbox{\footnotesize \bf (f) \scriptsize $(U_{{\rm num}} - U_{1})$({\color{black} \bf ---}) \  $(U_{{\rm num}} - U_{2})$({\color{blue} \bf - -})   }
\end{array}$
\end{center}
\vspace{-3mm}
\caption{\small Comparison of the weakly nonlinear solutions $U_1$ and $U_2$ with the numerical solution $U_{\rm num}$ for $k=1/\sqrt{3}$, $\epsilon=0.025$, $\gamma=0.1$
at (a) $t = 1 \ \& \ $(b) $t = 1/\epsilon$, with the close-up of some areas (c) \& (d) and the error plots (e) \& (f) at the respective times.
Numerical parameters: $\Delta t = 0.01$, $\Delta T = 0.000125$ and $L=80$, $N=4\times10^4$.}
 \label{figure:WNS_num_BOE}
\end{figure}

Figure \ref{figure:error_eps_BOE_T=1} displays double log plots of
the maximum absolute error found explicitly from simulations and the log of the
absolute maximum error as defined by (\ref{eqn:WNS_pert_num_error_rewritten}).
As one can see, the leading order maximum absolute error scales as $\mathcal{O}(\epsilon)$ for $U_1$
and as $\mathcal{O}(\epsilon^2)$ for $U_2$. This scaling agrees with the error estimates in Theorem \ref{theorem-main-2}.
It must be noted, however,  that the constant $C_2$ of the higher-order error term
(\ref{eqn:WNS_pert_num_error_rewritten}) is substantially larger than the corresponding
leading order error term constant $C_1$, and both constants are larger than in the case of the Boussinesq equation (see Table \ref{tab:error_scale}).  Therefore, although the error term scales as
$\mathcal{O}(\varepsilon^2)$, as expected, it is substantially larger than in the case of $\gamma = 0$
(compare the figures
 \ref{figure:WNS_BE_phi_num} (f) and  \ref{figure:WNS_num_BOE} (f) and note the difference in the values of $\epsilon$).

%
%\begin{figure}[h]
%\begin{center}$
%\begin{array}{c}
%\quad t=1/\epsilon  \\
%
%
%\includegraphics[width=2.7in]{23_07_13_BOE_gam=1_eps_vs_err.pdf}\\
%
%\quad  \mbox{\bf \scriptsize ${\rm ln}(e^{1}_{ta})$({\color{blue} \bf ---}) \ \ \   $\alpha_{1} {\rm ln}(\epsilon) + {\rm ln}(C_{1})$({\color{KMgreen} \bf ---}) \ \ \ ${\rm ln}(e^
%{2}_{ta})$({\color{red} \bf ---}) \ \ \  $\alpha_2 {\rm ln}(\epsilon) + {\rm ln}(C_2)$({\color{Turquoise} \bf ---})   } \\
%
%\end{array}$
%%\end{center}
%\caption{\small Double log plot of absolute errors versus $\epsilon$
%at $t =t_a=1/\epsilon \ (T=1)$.
%Coefficients are ${ \alpha_1=0.9754}$, ${\alpha_2=1.9753}$ and $C_1=13.90$, $C_2=2188$.
%Numerical parameters: $\Delta t = 0.01$, $\Delta T = 0.00125$ and $L=40$, $N=2\times10^4$.}
%\label{figure:error_eps_T=1_BOE_comparison}
%\end{figure}
%
\begin{figure}[h]
\begin{center}$
\begin{array}{cc}
\quad t=1 & \quad t=1/\epsilon  \\
\includegraphics[width=3.4in]{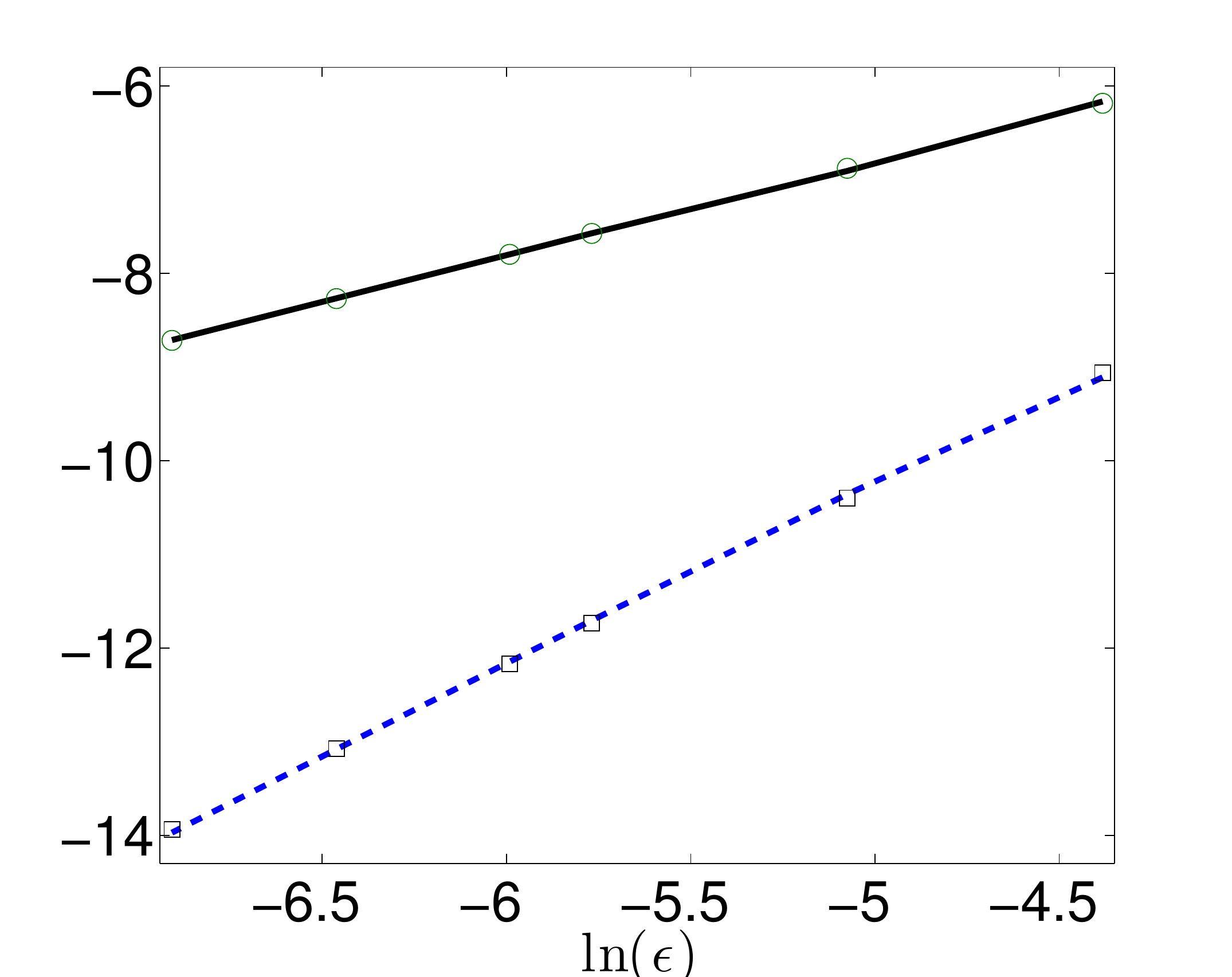}
&
\includegraphics[width=3.3in,height=2.75in]{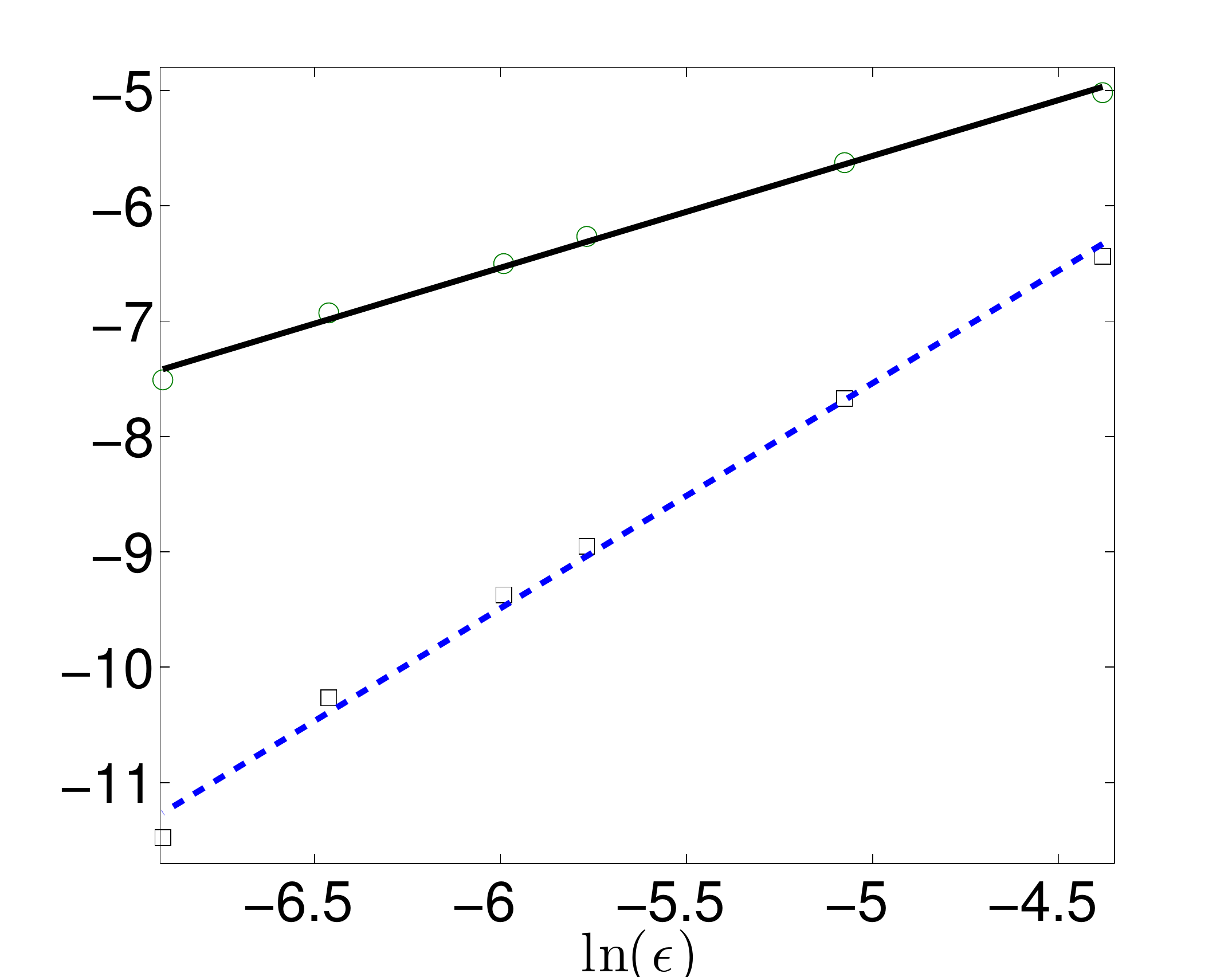}
\\
\multicolumn{2}{c}{  \quad  \mbox{\bf \scriptsize
${\rm ln}(e^{1}_{t})$({\color{black} \bf ---}) \ \ \ \  $\alpha_{1} {\rm ln}(\epsilon) + {\rm ln}(C_{1})$({\color{KMgreen}  $\bigcirc \ \bigcirc$}) \ \ \ \
${\rm ln}(e^{2}_{t})$({\color{blue} \bf - -}) \ \ \ \  $\alpha_{2} {\rm ln}(\epsilon) + {\rm ln}(C_{2})$({\color{black} $\square \ \square$})
} }
\end{array}$
\end{center}
\caption{\small Double log plot of absolute errors versus $\epsilon$
at $t =1$ and $t =1/\epsilon$. 
Numerical parameters are the same as in Figure \ref{figure:WNS_num_BOE}. 
Coefficients $\alpha_{1,2}$ and $C_{1,2}$ are
given in Table \ref{tab:error_scale}.}
\label{figure:error_eps_BOE_T=1}
\end{figure}

\section{Conclusions}

In this paper we constructed a weakly nonlinear solution of the Cauchy problem for the
regularized Boussinesq--Ostrovsky equation in the periodic domain. The weakly nonlinear
solution is constructed in terms of  solutions of two uncoupled Ostrovsky equations,
extending the results obtained, within the accuracy of the {\it physical} problem formulation,
in \cite{Karima1, Karima2}. In our present paper, it was shown how the accuracy of the
weakly nonlinear solution can be improved by using the linearized Ostrovsky equations 
for the two functions present in the first-order correction term. Although our consideration exceeded the accuracy
of the physical problem formulation because the second-order corrections terms were not included
in the original Boussiness--Ostrovsky equation, the methodology can be generalized and
extended to the case when the main equation includes these higher order terms.

Analytical results have been illustrated numerically, for the regularized Boussinesq equation
in the large domain, and the regularized Boussinesq--Ostrovsky equation in the moderate domain.
The behaviour of the error terms has been studied in details and the numerical results have
shown excellent agreement with the analytical predictions.

\bigskip

{\bf Acknowledgement:} D.P. appreciates support and hospitality of the Department of
Mathematical Sciences of Loughborough University. The research was supported by
the LMS Scheme 2 grant  and by the Loughborough University School of Science small grant.
K.K. and K.M. thank C. Klein and D. Tseluiko for helpful discussions of the numerical methods.

\end{document}